\newcommand{\longversion}{true}
\documentclass[letterpaper]{article}
\usepackage{aaai}
\usepackage{times}
\usepackage{helvet}
\usepackage{courier}

\usepackage{ifthen}

\usepackage{color}
\usepackage{url}

\usepackage{amsmath,amssymb}
\usepackage{latexsym}
\usepackage{xspace}
\usepackage{graphicx} %

\title{Strong Equivalence of Qualitative Optimization Problems}
\author{Wolfgang Faber\\
University of Calabria\\
{\rm faber@mat.unical.it}\\
\And
Miroslaw Truszczynski\\
University of Kentucky\\
{\rm mirek@cs.uky.edu}\\
\And
Stefan Woltran\\
Vienna University of Technology\\
{\rm woltran@dbai.tuwien.ac.at}}
\begin{document}
\maketitle

\def\prop#1#2#3{
\noindent
{\bf Proposition~\ref{#1}{#2}} {\it{#3}}}

\def\thrm#1#2#3{
\noindent
{\bf Theorem~\ref{#1}{#2}} {\it{#3}}}

\def\lemm#1#2#3{
\noindent
{\bf Lemma~\ref{#1}{#2}} {\it{#3}}}

\newcommand{\nop}[1]{}
\newcommand{\citeN}[1]{\citeauthor{#1} \shortcite{#1}}

\newcommand{\card}[1]{\ensuremath{|#1|}}

\newcommand{\powerset}[1]{\ensuremath{2^{#1}}}

\newcommand{\pair}[2]{(#1,#2)}

\newcommand{\Pol}{\ensuremath{\rm P}}
\newcommand{\NP}{\ensuremath{\rm NP}}
\newcommand{\CONP}{\mbox{\rm co-}\NP }
\newcommand{\SigmaP}[1]{\ensuremath{\Sigma_{#1}^P}}
\newcommand{\PiP}[1]{\ensuremath{{\Pi}_{#1}^{P}}}
\newcommand{\DeltaP}[1]{\ensuremath{\Delta_{#1}^P}}

\newcommand{\program}{\ensuremath{P}}
\newcommand{\qprogram}{\ensuremath{Q}}
\newcommand{\rprogram}{\ensuremath{R}}

\renewcommand{\P}{\program}
\newcommand{\Q}{\qprogram}
\newcommand{\R}{\rprogram}

\newcommand{\dff}{\ensuremath{\mathit{diff}}}
\newcommand{\hd}{\ensuremath{\mathit{hd}}}
\newcommand{\head}[1]{\ensuremath{\hd(#1)}}
\newcommand{\body}[1]{\ensuremath{\mathit{bd}(#1)}}
\newcommand{\pbody}[1]{\ensuremath{\mathit{body}^+(#1)}}
\newcommand{\nbody}[1]{\ensuremath{\mathit{body}^-(#1)}}
\newcommand{\GR}[2]{\ensuremath{R_{#1}(#2)}}
\newcommand{\LPif}{\leftarrow}
\newcommand{\la}{\LPif}

\newcommand{\nafo}[0]{\ensuremath{\mathit{not}}}
\newcommand{\naf}[0]{\nafo\;}
\newcommand{\reduct}[2]{\ensuremath{#1^{#2}}}
\newcommand{\CnO}[0]{\ensuremath{\mathit{Cn}}}
\newcommand{\Cn}[1]{\ensuremath{\CnO(#1)}}

\newcommand{\ASsym}[0]{\mathit{AS}}
\newcommand{\AS}[1]{\ASsym(#1)}

\newcommand{\Mod}[1]{\ensuremath{\mathit{Mod(#1)}}}
\newcommand{\Md}{\mathit{Mod}}

\newcommand{\SE}[1]{\ensuremath{\mathit{SE}(#1)}}
\newcommand{\UE}[1]{\ensuremath{\mathit{UE}(#1)}}
\newcommand{\UEp}[2]{\ensuremath{\mathit{UE}^{#1}(#2)}}
\newcommand{\HT}[1]{\ensuremath{\mathit{Mod_{HT}}(#1)}}

\newcommand{\atoms}[1]{\ensuremath{\mathit{atoms}(#1)}}
\newcommand{\progs}[1]{\ensuremath{\mathit{progs}(#1)}}
\newcommand{\facts}[1]{\ensuremath{\mathit{facts}(#1)}}

\newcommand{\pref}{\ensuremath{\pi}}
\newcommand{\vph}{\ensuremath{\varphi}}
\newcommand{\acc}{\ensuremath{\mu}}

\newcommand{\gensym}{\ensuremath{g}}
\newcommand{\prefsym}{\ensuremath{<}}

\newcommand{\pprogram}{\ensuremath{{\cal P}}}
\newcommand{\qpprogram}{\ensuremath{{\cal Q}}}
\newcommand{\rpprogram}{\ensuremath{{\cal R}}}

\newcommand{\PP}{\pprogram}
\newcommand{\QQ}{\qpprogram}
\newcommand{\RR}{\rpprogram}

\newcommand{\pprogramgen}[1]{\ensuremath{{#1}^{\gensym{}}}}
\newcommand{\pprogrampref}[1]{\ensuremath{{#1}^{\prefsym}}}
\newcommand{\pprogrampair}[2]{\ensuremath{\pair{#1}{#2}}}
\newcommand{\pprogrampairbase}[1]{\ensuremath{\pprogrampair{\pprogramgen{#1}}{\pprogrampref{#1}}}}

\newcommand{\Pgen}{\pprogramgen{\P}}
\newcommand{\Ppref}{\pprogrampref{\P}}
\newcommand{\Qgen}{\pprogramgen{\Q}}
\newcommand{\Qpref}{\pprogrampref{\Q}}
\newcommand{\Rgen}{\pprogramgen{\R}}
\newcommand{\Rpref}{\pprogrampref{\R}}

\newcommand{\LPifrank}[1]{\stackrel{#1}{\leftarrow}}

\newcommand{\headi}[2]{\ensuremath{\hd_{#1}(#2)}}

\newcommand{\PAS}[1]{\ensuremath{\mathit{PAS}(#1)}}
\newcommand{\PS}{\ensuremath{\mathit{PAS}}}

\newcommand{\ASO}{\emph{ASO}\xspace}
\newcommand{\RASO}{\emph{RASO}\xspace}

\newcommand{\betterrel}[1]{\ensuremath{>_{#1}}}
\newcommand{\bettereqrel}[1]{\ensuremath{\geq_{#1}}}

\newcommand{\bettereq}[3]{\ensuremath{#1 \geq_{#3} #2}}
\newcommand{\better}[3]{\ensuremath{#1 >_{#3} #2}}
\newcommand{\nbetter}[3]{\ensuremath{#1 \not >_{#3} #2}}
\newcommand{\nbettereq}[3]{\ensuremath{#1 \not \geq_{#3} #2}}

\newcommand{\restr}[2]{\ensuremath{#1\mid_{#2}}}

\newcommand{\prefschema}{\ensuremath{\sigma}}

\newcommand{\equivk}{\equiv^\prefschema}

\newcommand{\rank}[1]{\ensuremath{\mathit{rank}(#1)}}

\newcommand{\leftsup}[2]{{\vphantom{#2}}^{#1}{#2}}

\newcommand{\degree}[2]{\ensuremath{v_{#1}(#2)}}

\newcommand{\sequivmacro}[1]{\ensuremath{\equiv_s^{#1}}}

\newcommand{\sequiv}{\sequivmacro{}}
\newcommand{\sequivgen}{\sequivmacro{\gensym}}
\newcommand{\sequivpref}{\sequivmacro{\prefsym}}

\newcommand{\meta}{\ensuremath{\;\widetilde{\equiv}\;}}

\renewcommand{\L}{{\cal L}_{{\cal U}}}
\newcommand{\Lpref}{{\cal L}^\mathit{pref}_{{\cal U}}}
\newcommand{\Lrpref}{{\cal L}^\mathit{rpref}_{{\cal U}}}
\newcommand{\Lgen}{{\cal L}^\mathit{gen}_{{\cal U}}}
\newcommand{\Laso}{{\cal L}^\mathit{aso}_{{\cal U}}}
\newcommand{\Lraso}{{\cal L}^\mathit{raso}_{{\cal U}}}

\newtheorem{definition}{Definition} %
\newtheorem{theorem}{Theorem} %
\newtheorem{proposition}[theorem]{Proposition}
\newtheorem{corollary}[theorem]{Corollary}
\newtheorem{lemma}[theorem]{Lemma}
\newtheorem{example}{Example}
\newtheorem{conjecture}{Conjecture}
\newtheorem{observation}{Observation}

\newenvironment{proof}[0]{\mbox{}{\em Proof.}%
 \ }{\ $\Box$\\}

\newcommand{\TODO}[1]{{\bf\large TODO:} {\em #1 }}

\newcommand{\natnum}{\ensuremath{\mathbb{N}}}

\newcommand{\intnum}{\ensuremath{\mathbb{Z}}}

\newcommand{\st}{\;|\;}

\newcommand{\ifof}{if and only if\ }

\newcommand{\U}{{\cal U}}
\renewcommand{\S}{{\cal S}}

\newcommand{\LongShort}[2]{\ifthenelse{\not\isundefined{\longversion}}{#1}{#2}}
\newcommand{\LongOnly}[1]{\LongShort{#1}{}}

\begin{abstract}
We introduce the framework of \emph{qualitative optimization problems}
(or, simply, optimization problems) to represent preference theories.
The formalism uses separate modules to describe the space of outcomes
to be compared (the \emph{generator}) and the preferences on outcomes
(the \emph{selector}). We consider two types of optimization problems.
They differ in the way the generator, which we model by a propositional
theory, is interpreted: by the standard propositional logic semantics,
and by the equilibrium-model (answer-set) semantics. Under the latter
interpretation of generators, optimization problems directly generalize
answer-set optimization programs proposed previously. We study
\emph{strong equivalence} of optimization problems, which guarantees
their interchangeability within any larger context. We characterize
several versions of strong equivalence obtained by restricting the
class of optimization problems that can be used as extensions and
establish the complexity of associated reasoning tasks. Understanding
strong equivalence is essential for modular representation of
optimization problems and rewriting techniques to simplify them
without changing their inherent properties.
\end{abstract}

\section{Introduction}
\label{sec:introduction}

We introduce the framework of \emph{qualitative optimization problems} 
in which, following the design of answer-set optimization (ASO) programs
\cite{bnt03}, we use separate modules to describe the space of outcomes
to be compared (the \emph{generator}) and the preferences on the 
outcomes (the \emph{selector}). In all optimization problems we 
consider, the selector module follows the syntax and the semantics of 
preference modules in ASO programs, and the generator is given by a 
propositional theory. If this propositional theory is interpreted 
according to the standard propositional logic semantics, that is, 
outcomes to be compared are models of the generator, we speak about the 
\emph{classical optimization problems} (CO problems, for short). If the 
generator theory is interpreted by the semantics of \emph{equilibrium} 
models \cite{Pearce97}, also known as \emph{answer sets} \cite{fer05},  
that is, it is the answer sets of the generator that are being compared, 
we speak about \emph{answer-set optimization problems} (ASO problems, 
for short).

Representing and reasoning about preferences in qualitative settings
is an important research area for knowledge representation and
qualitative decision theory. The main objectives are to design expressive
yet intuitive languages to model preferences, and to develop automated
methods to reason about formal representations of preferences in these
languages. The literature on the subject of preferences is vast. We refer
the reader to the special issue of Artificial Intelligence Magazine 
\cite{aim08} for a collection of overview articles and references.

Understanding when optimization problems are equivalent, in particular,
when one can be interchanged with another within any larger context, is
fundamental to any preference formalism. Speaking informally, optimization
problems $P$ and $Q$ are \emph{interchangeable} or \emph{strongly 
equivalent} when for every optimization problem $R$ (\emph{context}), 
$P\cup R$ and $Q\cup R$ define the same optimal models. Understanding
when one optimization problem is equivalent to another in this sense
is essential for preference analysis, modular preference representation, 
and rewriting techniques to simplify optimization problems into forms
more amenable to processing, without changing any of their inherent
properties. Let us consider a multi-agent setting, in which agents
combine their preferences on some set of alternatives with the goal of
identifying optimal ones. Can one agent in the ensemble be replaced
with another so that the set of optimal alternatives is unaffected not
only now, but also under any extension of the ensemble in the future?
Strong equivalence of agents' optimization problems is precisely what
is needed to guarantee this full interchangeability property! 

The notion of strong equivalence is of general interest, by no means
restricted to preference formalisms. In some cases, most notably for
classical logic, it coincides with \emph{equivalence}, the property of 
having the same models. However, if the semantics is not monotone, that
is, extending the theory may introduce new models, not only eliminate 
some, strong equivalence becomes a strictly stronger concept, and the
one to adopt if theories being analyzed are to be placed within a 
larger context. The nonmonotonicity of the semantics is the salient 
feature of nonmonotonic logics \cite{mt93} and strong equivalence of 
theories in nonmonotonic logics, especially logic programming with the 
answer-set semantics \cite{gelf-lifs-91}, was extensively studied in 
that setting \cite{lpv01,tu03,efw04}. Preference formalisms also often 
behave nonmonotonically as adding a new preference may cause a 
non-optimal outcome (model) to become an optimal one. Thus, in 
preference formalisms, equivalence and strong equivalence are typically 
different notions. Accordingly, strong equivalence was studied for logic
programs with rule preferences~\cite{FaberK06}, programs with ordered 
disjunction~\cite{FaberTW08} and %
programs with weak 
constraints~\cite{EiterFFW07}.

We 
extend the study of strong equivalence to the formalism 
of \emph{qualitative optimization problems}.
The formalism is motivated by the design of answer-set optimization 
(ASO) programs of \citeauthor{bnt03}~\shortcite{bnt03}. It borrows two
key features from ASO programs that make it an attractive alternative
to the preference modeling approaches based on logic programming that we
mentioned above. First, following ASO programs, optimization problems 
provide a clear separation of hard constraints, which specify the space
of feasible outcomes, and preferences (soft constraints) that impose a 
preference ordering on feasible outcomes. Second, optimization problems 
adopt the syntax and the semantics of preference rules of ASO programs
that correspond closely to linguistic patterns of simple conditional 
preferences used by humans.

The separation of preference modules from hard-constraints
facilitates eliciting and representing preferences. 
It is also important for characterizing strong 
equivalence.  When a clear separation is not present, 
like in
logic programs with ordered disjunctions, strong equivalence 
characterizations are cumbersome as they have to account for complex 
and mostly implicit interactions between hard constraints and 
preferences. For optimization problems, which impose the separation, 
we have ``one-dimensional'' forms of strong equivalence, in which 
only hard constraints or only preferences are added. These 
``one-dimensional'' concepts are easier to study yet provide enough 
information to construct characterizations for the general case. 

\smallskip
\noindent
{\bf Main Contributions.}
(1) We propose a general framework of qualitative optimization problems,
extending 
in several ways the formalism of ASO programs. We focus on two important
instantiations of the framework, the classes of classical optimization 
(CO) problems and answer-set optimization (ASO) problems. The latter one 
directly generalizes ASO programs. 
(2) We characterize the concept of strong equivalence of optimization
problems relative to changing selector modules. The characterization
is independent of the semantics of generators and so, applies both to
CO and ASP problems. We also characterize strong equivalence relative
to changing generators (with preferences fixed). In this case, not
surprisingly, the characterization depends on the semantics of 
generators. However, we show that the dependence is quite uniform, and 
involves a characterization of strong equivalence of generators relative
to their underlying semantics, when they are considered on their own as
propositional theories. Finally, we combine the characterizations of 
the ``one-dimensional'' concepts of strong equivalence into a 
characterization of the general ``combined'' notion. 
(3) We develop our results for the case when preferences are ranked. In 
practice, preferences are commonly ranked due to the hierarchical structure 
of preference providers. The general case we study allows for additions
of preferences of ranks from a specified interval $[i,j]$. This covers
the case when only some segment in the hierarchy of preference providers
is allowed to add preferences (top decision makers, middle management,
low-level designers), as well as the case when there is no distinction 
between the importance of preferences (the non-ranked case). (4) We 
establish the complexity of deciding whether two optimization problems
are strongly equivalent relative to changing selectors, generators, or
both. 
\LongShort{We present only proof sketches and some of the simpler and
  not overly technical proofs in the main text. Detailed proofs can be
  found in the Appendix.}{Due to the space restriction, we present
  only proof sketches and some of the simpler and not overly technical
  proofs here. All proofs can be found in the full version of the
  paper at \url{http://arxiv.org/abs/??????????}.}

\section{Optimization Problems}\label{sec:background}

A \emph{qualitative optimization problem} (an \emph{optimization
  problem}, from now on) is an ordered pair $P=(T,S)$, where $T$ is
called the \emph{generator} and $S$ the \emph{selector}. The role of
the generator is to specify the family of \emph{outcomes} to be
compared. The role of the selector $S$ is to define a relation $\geq$
on the set of outcomes and, consequently, define the notion of an
optimal outcome. The relation $\geq$ induces relations $>$ and
$\approx$: we define $I> J$ if $I\geq J$ and $J\not\geq I$, and $I\approx J$ if
$I\geq J$ and $J\geq I$. For an optimization problem $P$, we write
$P^g$ and $P^s$ to refer to its generator and selector, respectively.

\smallskip
\noindent
\textbf{Generators.}
As generators we use propositional theories in the language determined
by a fixed countable universe $\U$, a boolean constant $\bot$, and boolean
connectives $\land$, $\lor$ and $\rightarrow$, and where we define the
constant $\top$, and the connectives $\neg$ and $\leftrightarrow$ in 
the usual way.\footnote{While the choice of primitive connectives is not 
common for the language of classical propositional logic, it is standard
for the of logic here-and-there which underlies the answer-set semantics.}
Models of the generator, as defined by the semantics 
used, represent \emph{outcomes} of the corresponding optimization 
problem.  We consider two 
quite different 
semantics for generators: 
the classical propositional logic semantics and the semantics of 
\emph{equilibrium} models \cite{Pearce97}. Thus, outcomes are either
models or equilibrium models, depending on the semantics chosen.
The first semantics is of interest due to the fundamental role and 
wide\-spread use of classical propositional logic, in particular, as 
a means to describe constraints. Equilibrium models generalize answer 
sets of logic programs to the case of arbitrary propositional theories 
\cite{Pearce97,fer05} and are often referred to as answer sets. The
semantics of equilibrium models is important due to the demonstrated 
effectiveness of logic programming with the semantics of answer sets 
for knowledge representation applications. We use the terms equilibrium
models and answer sets interchangeably. 

Throughout the paper, 
we represent interpretations of $\U$
as subsets of $\U$. We write $I\models \phi$ to state that an 
interpretation $I\subseteq \U$ is a (classical propositional) model 
of a formula $\phi$. Furthermore, we denote the set of classical models 
of a formula or theory $T$ by $\Md(T)$.

Equilibrium models arise in the context of the propositional logic of 
here-and-there, or the logic HT for short \cite{ht30}. We briefly recall
here definitions of concepts, as well as properties of the logic HT that are 
directly relevant to our work. We refer to the papers by 
\citeauthor{Pearce97}~\shortcite{Pearce97} and 
\citeauthor{fer05}~\shortcite{fer05} for further details.

The logic HT is a logic located between the intuitionistic and the 
classical logics. Interpretations in the logic HT are pairs 
$\langle I,J\rangle$ of 
standard propositional interpretations such that $I\subseteq J$. We 
write $\langle I,J\rangle \models_{HT} \phi$ to denote that a formula 
$\phi$ holds in an interpretation $\langle I,J\rangle$ in the logic HT. 
The relation $\models_{HT}$ is defined recursively as follows: for an
atom $a$, $\langle I,J\rangle \models_{HT} a$ if and only if $a\in I$. 
The cases of the boolean connectives $\land$ and $\lor$ are standard, 
and $\langle I,J\rangle \models_{HT} \vph\rightarrow\psi$ if and only 
if $J\models \vph\rightarrow\psi$ (classical satisfiability) and $\langle
I,J\rangle \not\models_{HT} \vph$ or $\langle I,J\rangle \models_{HT}
\psi$. Finally, $\langle I,J\rangle \not \models_{HT} \bot$. %

An \emph{equilibrium model} or \emph{answer set} of a propositional 
theory $T$ is a standard interpretation $I$ such that 
$\langle I,I\rangle\models_{HT} T$ and for 
every proper subset $J$ of $I$, $\langle J,I \rangle\not \models_{HT} 
T$. Answer sets of a propositional theory $T$ are also classical
models of $T$. The converse is not true in general.

We denote the set of all answer sets of a theory 
$T$ by $\AS{T}$, and the set of all HT-models of $T$ by $\HT{T}$, that 
is, $\HT{T} =\{ \langle I,J\rangle \mid I\subseteq J, \langle I,J\rangle 
\models_{HT} T\}$. 

For each of the semantics there are two natural concepts of equivalence.
Two theories $T_1$ and $T_2$ are \emph{equivalent} if they have the same
models (classical or equilibrium, respectively). They are 
\emph{strongly equivalent} if for every theory $S$, $T_1\cup S$ and 
$T_2\cup S$ have the same models (again, classical or equilibrium, 
respectively).

For classical semantics, strong equivalence and equivalence 
coincide. It is not so for the semantics of equilibrium models. The 
result by \citeauthor{lpv01}~\shortcite{lpv01} states that 
two theories $T_1$ and $T_2$ are strongly equivalent for equilibrium 
models if and only if $T_1$ and $T_2$ are equivalent in the logic HT, 
that is $\HT{T_1} = \HT{T_2}$.

We call optimization problems under the classical interpretation of 
generators \emph{classical optimization problems} or CO problems, for 
short. When we use the answer-set semantics for generators, we
speak about \emph{answer-set optimization problems} or ASO problems, 
for short.

\smallskip
\noindent
\textbf{Selectors.}
We follow the definitions of preference modules in ASO
programs \cite{bnt03}, 
adjusting the terminology to our more
general setting. A \emph{selector} is a finite set of \emph{ranked
preference rules} 
\begin{equation}\label{eq::rprefrule}
\phi_1 > \cdots > \phi_k \LPifrank{j} \psi
\end{equation}
where $k$ and $j$ are positive integers, and $\phi_i$, $1\leq i\leq k$, 
and $\psi$ are propositional formulas over $\U$. For a rule $r$ of the 
form (\ref{eq::rprefrule}), the number $j$ is the \emph{rank} of $r$,
denoted by $\rank{r}$, $\head{r} = \{\phi_1, \ldots, \phi_k\}$ is the 
\emph{head} of $r$ and $\psi$ is the \emph{body} of $r$, $\body{r}$. 
Moreover, we 
write $\headi{i}{r}$ to refer to formula $\phi_i$.

If $\rank{r}=1$ for every preference rule $r$ in a selector $S$, then
$S$ is a \emph{simple selector}. Otherwise, $S$ is \emph{ranked}. We
often omit ``$1$'' from the notation ``$\LPifrank{1}$'' for simple
selectors.  For a selector $S$, and $i,j\in\{0,1,2,\ldots\} \cup
\{\infty\}$, we define $S_{[i,j]}=\{ r\in S \mid i\leq \rank{r} \leq
j\}$ (where we assume that for every integer $k$, $k<\infty$) and
write $[i,j]$ for the \emph{rank interval} $\{k \mid i \leq k \leq
j\}$. We extend this notation to optimization problems. For $P=(T,S)$
and a rank interval $[i,j]$, we set $P_{[i,j]} = (T,S_{[i,j]})$. For
some rank intervals we use shorthands, for example $=i$ for $[i,i]$,
$<i$ for $[1,i-1]$, $\geq i$ for $[i,\infty]$, and similar.

For an interpretation $I$, a \emph{satisfaction degree} of a preference
rule $r$ is $\degree{I}{r}=min\{i\mid I \models \headi{i}{r}\}$, if
$I\models \body{r}$ and $I\models %
\bigvee%
\head{r}$; otherwise, the rule is
\emph{irrelevant} to $I$, and $\degree{I}{r}=1$. 
We note that 
\citeauthor{bnt03} (\citeyear{bnt03}) represented the satisfaction 
degree of an irrelevant rule by a special non-numeric degree, treated 
as being equivalent to $1$. The difference is immaterial and the two 
approaches are equivalent.

Selectors determine a preference relation on interpretations. Given
interpretations $I$ and $J$ and a simple selector $S$, $I\geq^S J$
holds precisely when for all $r\in S$, $\degree{I}{r}\leq\degree{J}{r}$. Therefore, $I >^S J$ holds if and only if $I\geq^S J$ and there exists $r\in S$
such that $\degree{I}{r} < \degree{J}{r}$; $I \approx^S J$ if
and only if for every $r\in S$, $\degree{I}{r} = \degree{J}{r}$.

Given a \emph{ranked} selector $S$, we define $I\geq^S J$ if for every 
preference rule $r\in S$, $\degree{I}{r} =\degree{J}{r}$, or if there 
is a rule $r'\in S$ such that the following three conditions hold:
\begin{enumerate}
\item $\degree{I}{r'} < \degree{J}{r'}$  
\item for every $r\in S$ of the same rank as $r'$, 
$\degree{I}{r} \leq \degree{J}{r}$
\item for every $r\in S$ of smaller rank than $r'$, 
$\degree{I}{r} = \degree{J}{r}$. 
\end{enumerate}
Moreover, $I>^S J$ if and only if there is a rule $r'$ for which the
three conditions above hold, and $I \approx^S J$ if and only if for
every $r\in S$, $\degree{I}{r} = \degree{J}{r}$. Given an optimization
problem $P=(T,S)$ we often write $\geq^P$ for $\geq^S$ (and similarly
for $>$ and $\approx$).

\smallskip
\noindent
\textbf{Optimal (preferred) outcomes.} For an optimization problem $P$,
$\acc(P)$ denotes the set of all outcomes of $P$, that is, the set of 
all models (under the selected semantics) of the generator of $P$. 
Thus, $\mu(P)$ stands for all models of $P$ in the framework of CO 
problems and for all answer sets of $P$, when ASO problems are 
considered. A model $I \in \acc(P)$ is \emph{optimal} or \emph{preferred}
for $P$ if there is no model $J\in \acc(P)$ such that $J>^P I$. We denote 
the set of all \emph{preferred models} of $P$ by $\pref(P)$. 

\smallskip
\noindent
\textbf{Relation to ASO programs.} Optimization problems 
extend the 
formalism
of ASO programs \cite{bnt03}
in several 
ways.
First, the generator programs are 
\emph{arbitrary} propositional theories. Under the semantics of 
equilibrium models, our generators properly extend logic programs with 
the answer-set semantics used as generators in ASO programs. Second, 
the selectors use arbitrary propositional formulas for options in the 
heads of preference rule, as well as for conditions in their bodies. 
Finally, optimization problems explicitly allow for alternative 
semantics of generators, a possibility mentioned but not pursued by 
\citeauthor{bnt03}~\shortcite{bnt03}. 

\smallskip
\noindent
\textbf{Notions of Equivalence.}
We define the union of optimization problems as expected, that is, for
$P_1=(T_1,S_1)$ and $P_2=(T_2,S_2)$, we set $P_1\cup P_2=(T_1\cup 
T_2,S_1\cup S_2)$. Two optimization problems $P_1$ and $P_2$ are 
\emph{strongly equivalent} with respect to a class ${\cal R}$ of 
optimization problems (\emph{contexts}) if for every optimization 
problem $R\in {\cal R}$, $\pi(P_1\cup R) = \pi(P_2\cup R)$.

We consider three general classes of contexts. First and foremost, we 
are interested in the class $\L$ of all optimization problems over $\U$.
We also consider the families $\L^g$ and $\L^s$ of all optimization 
problems of the form $(T,\emptyset)$ and $(\emptyset,S)$, respectively. 
The first class consists of optimization problems where all models of 
the generator are equally preferred. We call such optimization problems 
\emph{generator problems}. The second class consists of optimization 
problems in which every interpretation of $\U$ is an acceptable outcome.
We call such optimization problems \emph{selector problems}. These 
``one-dimensional'' contexts provide essential insights into the general 
case. For the first class, we simply speak of \emph{strong equivalence}, 
denoted $\equiv^s_g$. For the latter two classes, we speak of 
\emph{strong gen-equivalence}, denoted $\equiv_g$, and \emph{strong 
sel-equivalence}, denoted $\equiv^s$, respectively.

Constraining ranks of rules in selectors gives rise to two additional
classes of contexts parameterized using rank intervals $[i,j]$:
\begin{enumerate}
\item $\L^{s,[i,j]}=\{(\emptyset,S)\in\L^s\mid  S=S_{[i,j]}\}$
\item $\L^{[i,j]}=\{(T,S)\in \L\mid S=S_{[i,j]}\}$
\end{enumerate}

The first class of contexts gives rise to \emph{strong sel-equivalence} 
with respect to rules of rank in $[i,j]$, denoted by $\equiv^{s,[i,j]}$.
The second class of contexts yields the concept of 
\emph{strong equivalence} with respect to rules of rank in $[i,j]$.
We denote it by $\equiv^{s,[i,j]}_g$. We call problems in the class 
$\L^{=1}=\L^{[1,1]}$ \emph{simple optimization problems}.

\section{Examples} 

We present now examples that illustrate key issues relevant to 
strong equivalence of optimization problems. %
They point to the necessity of some conditions that appear later in 
characterizations of strong equivalence and hint at some constructions 
used in proofs.
In all examples except for
the last one, we consider simple CO problems. 
In all problems only atoms explicitly listed matter, so 
we disregard all others. 

\begin{example}
\label{ex:1n}
Let $P_1=(T_1,S_1)$, where $T_1=\{a\leftrightarrow \neg b\}$ and 
$S_1=\{a > b \LPif\}$. There are two outcomes here, $\{a\}$ and $\{b\}$,
that is, $\mu(P_1)=\{\{a\},\{b\}\}$. Let $r$ be the only preference 
rule in $S_1$. Clearly, $v_{\{a\}}(r)=1$ and $v_{\{b\}}(r)=2$.  Thus, 
$\{a\} >^{P_1} \{b\}$ and so, $\pi(P_1)=\{\{a\}\}$.

In addition, let $P_2=(T_2,S_1)$, where $T_2=\{a\land \neg b\}$ and 
$S_1$ is as above. Then, $\mu(P_1)=\{\{a\}\}$ and, trivially, 
$\pi(P_1)=\{\{a\}\}$. It follows that $P_1$ and $P_2$ are equivalent, 
as they specify the same optimal outcomes. However, they are not 
strongly gen-equivalent (and so, also not strongly equivalent). Indeed, 
let $R=(\{\neg a\},\emptyset)$. Then $\mu(P_1\cup R) =\{\{b\}\}$ and so, 
$\pi(P_1\cup R)=\{\{b\}\}$. On the other hand, $\mu(P_2\cup R)=\emptyset$ 
and, therefore, $\pi(P_2\cup R)=\emptyset$. 
\end{example}

Example \ref{ex:1n} 
suggests that we must have
$\mu(P_1)=\mu(P_2)$ if problems $P_1$ and $P_2$ are to be strongly 
\mbox{(gen-)}equivalent. Otherwise, by properly selecting the context 
generator, we can eliminate all outcomes in one problem still leaving 
some in the other. 

\begin{example}
\label{ex:2n}
Let $P_3=(T_3,S_3)$, where $T_3=\{a\vee b\vee c, \neg(a\land b),
\neg(a\land c), \neg(b\land c)\}$ and $S_3=\{a > b \LPif,\ a> c \LPif\}$.
We have $\mu(P_3)=\{\{a\},\{b\},\{c\}\}$. In addition, $\{a\} >^{P_3} \{b\}$,
$\{a\} >^{P_3} \{c\}$, and $\{b\}$ and $\{c\}$ are incomparable. Thus,
$\pi(P_3)=\{\{a\}\}$. 
Let now $P_4=(T_4,S_4)$, where $T_4=T_3$ and
$S_4=\{a>b>c \LPif\}$. Clearly, $\mu(P_4)=\mu(P_3)=\{\{a\},\{b\},\{c\}\}$.
Moreover, $\{a\}>^{P_4}\{b\}>^{P_4}\{c\}$. Thus, $\pi(P_4)=\{\{a\}\}$ and so,
$P_3$ and $P_4$ are equivalent. They are not strongly (gen-)equivalent.
Indeed, let $R=(\{\neg a\},\emptyset)$. 
Then, $\pi(P_3\cup R)=
\{\{b\},\{c\}\}$ and that $\pi(P_4\cup R)=\{\{b\}\}$. 
\end{example}

This example suggests
that for two optimization problems to be strongly (gen-)equivalent, they
have to define the same preference relation $>$ on outcomes.

\begin{example}
\label{ex:3n}
Let $P_5=(T_5,S_5)$, where $T_5=\{a,\neg b\}$ and $S_5=\emptyset$.
We have $\mu(P_5)=\{\{a\}\}$ and so, $\pi(P_5)=\{\{a\}\}$. It follows that
$P_5$ is equivalent to $P_1$. Let $R=(\emptyset, \{b>a\LPif\})$. Since
$\mu(P_5\cup R)=\mu(P_5)=\{\{a\}\}$, $\pi(P_5\cup R)=\{\{a\}\}$. Further,
$P_1\cup R=(T_1, \{a>b\LPif, b>a\LPif\})$ and so, $\mu(P_1\cup R)=\mu(P_1)
=\{\{a\},\{b\}\})$ and we get $\pi(P_1\cup R)=\{\{a\},\{b\}\}$. Thus, $P_5$ and
$P_1$ are not strongly (sel-)equivalent. 
\end{example}

Informally, this example shows that by modifying the selector part, we 
can make non-optimal outcomes optimal. Thus, as in the case of strong 
gen-equivalence (Example \ref{ex:1n}) the equality of sets of 
models (i.e., equivalence) is important 
for strong \mbox{(sel-)}equivalence (a more refined condition will be 
needed for ranked programs, as we show in Theorem 
\ref{thm:sequivrankedprefeq}).

\begin{example}
\label{ex:4n}
For the next example, let us consider problems $P_6=(T_1,S_6)$ and
$P_7=(T_1,S_7)$, where $T_1$ is the generator from Example \ref{ex:1n}, 
$S_6=\{a >b \LPif,\ b>a\LPif\}$ and $S_7=\emptyset$. We have
$\mu(P_6)=\mu(P_7)= \{\{a\},\{b\}\}$. Moreover, $\{a\}\geq^{P_6} \{b\}$ 
and $\{b\}\geq^{P_6} \{a\}$. Thus, $\pi(P_6)=\{\{a\},\{b\}\}$. Since 
$P_7^s=\emptyset$, we also have (trivially) that $\{a\}\approx^{P_7}
\{b\}$. Thus, $\pi(P_7)=\{\{a\},\{b\}\}$, too, and the problems $P_6$
and $P_7$ are equivalent. They are not strongly sel-equivalent, though.
Let $R=(\emptyset,\{a>b \LPif\})$. Then, $P_6\cup R=P_6$ and so, 
$\pi(P_6\cup R)=\{\{a\},\{b\}\}$. On the other hand, $\{a\}>^{P_7\cup R}
\{b\}$. Thus, $\pi(P_7\cup R)=\{\{a\}\}$. 
\end{example}

The example suggests that 
for strong sel-equivalence
the 
equality of the relation $\geq$ induced by the problems 
considered 
is important. The equality of the relation $>$ is not sufficient. In our 
example, the relations $>^{P_6}$ and $>^{P_7}$ are both empty and so 
-- equal. But they are empty for different reasons, absence of preference 
versus conflicting preferences, which can give rise to different preferred 
models when extending selectors. 

Our last example involves ranked problems. It is meant to hint at
issues that arise when ranked problems are considered. In the 
general case of ranked selectors, the equality of the relation $>$ 
induced by programs being evaluated is not related to strong sel-equivalence
in any direct way and an appropriate modification of that requirement has
to be used together with yet another condition (cf.\ 
Theorem~\ref{thm:sequivrankedprefeq}).

\begin{example}
\label{ex:6n}
Let $P=(T_1,S)$, where 
$S=\{a>b\LPifrank{2}\}$ and $P'=(T_1,S')$, where 
$S'=\{a>b\LPifrank{3}\}$. Clearly, $\mu(P)=\mu(P')=\{\{a\},\{b\}\}$
and $\pi(P)=\pi(P')=\{a\}$. Thus, the two problems are equivalent.
They are not strongly sel-equivalent if arbitrary selectors are allowed.
For instance, let $R=(\emptyset, \{b>a\LPifrank{2}\})$. Adding this new
preference rule to $P$ makes $\{a\}$ and $\{b\}$ incomparable and so,
$\pi(P\cup R)=\{\{a\},\{b\}\}$. On the other hand, since the new rule has
rank 2, it dominates the preference rule of $P'$, which is of rank 3.
Thus, $\mu(P'\cup R)=\{\{b\}\}$. A similar effect occurs with the problem $R'=
(\emptyset, \{b>a\LPifrank{3}\})$. Since its only preference rule
is dominated by the only preference rule in $P$,  
$\pi(P\cup R')= \pi(P)= \{\{a\}\}$.
On the other hand, $\{a\}$ and $\{b\}$ are incomparable in $P'\cup R'$
and, consequently, $\pi(P'\cup R')=\{\{a\},\{b\}\}$. Thus, extending $P$
and $P'$ with selectors containing rules of rank 2 or 3 may lead to 
different optimal outcomes. It is so even though
the relations $>$ induced by $P$ and $P'$
on the set of all outcomes coincide. 

However, 
adding selectors consisting only of rules
of rank greater than 3 cannot have such an effect, since the existing
rules would dominate them. Also, adding rules of rank 1 cannot result
in differing preferred answer sets, as such rules would dominate the
existing ones. Formally, the problems $P$ and $P'$ are strongly
sel-equivalent relative to selectors with preference rules of rank
greater than 3 or less than 2.
\end{example}
 
\section{Strong sel-equivalence}
\label{sec:str-sel-eq}

We start with the case of strong sel-equivalence, 
the core case for our study. Indeed, characterizations of strong 
sel-equivalence naturally imply characterizations for the general case 
thanks to the following simple observation.

\begin{proposition}
\label{prop:ortho}
Let $P$ and $Q$ be optimization problems (either under classical or
answer-set semantics for the generators) and $[i,j]$ a rank
interval. Then $P\equiv^{s,[i,j]}_g Q$ if and only if for every
generator $R\in \L^g$, $P\cup R\equiv^{s,[i,j]} Q\cup R$.
\end{proposition}
\begin{proof}
($\Rightarrow$)
Let $R\in \L^g$. Since $P\equiv^{s,[i,j]}_g Q$,
$P\cup R\equiv^{s,[i,j]}_g Q\cup R$ and so, $P\cup R\equiv^{s,[i,j]}
Q\cup R$.

\smallskip
\noindent
($\Leftarrow$) Let $R$ be any optimization problem in $\L^{[i,j]}$. We have $P\cup R=(P\cup (R^g,\emptyset))\cup
(\emptyset,\R^s)$ and $Q\cup R=(Q\cup (R^g,\emptyset))\cup
(\emptyset,\R^s)$. By the assumption, it follows that $P\cup (R^g,\emptyset)
\equiv^{s,[i,j]} Q\cup (R^g,\emptyset)$. Thus,
\[ \pref((P\cup (R^g,\emptyset))\cup (\emptyset,\R^s))=
\pref((Q\cup (R^g,\emptyset))\cup (\emptyset,\R^s)). \]
It follows that $\pref(P\cup R)=\pref(Q\cup R)$ and, consequently,
that $P\equiv^{s,[i,j]}_g Q$.
\end{proof}

Furthermore, %
the set of outcomes of an optimization problem 
$P$ is unaffected by changes in the selector module. It follows that
the choice of the semantics for generators does not matter 
for characterizations of strong sel-equivalence. Thus, 
whenever in this section we refer to the set of outcomes of 
an optimization problem $P$, we use the notation $\mu(P)$, and 
not the more specific one, $\Mod{P^g}$ or $\AS{P^g}$, that applies
to CO and ASO problems, respectively.

Our first main result concerns strong sel-equivalence relative to
selectors consisting of preference rules of ranks in a rank interval
$[i,j]$. Special cases for strong sel-equivalence will follow as
corollaries.
To state the result, we
need some auxiliary notation. For an optimization problem $P$, we
define $\dff^P(I,J)$ to be the largest $k$ such that $I
\approx^{P_{<k}} J$. If for every $k$ we have $I \approx^{P_{<k}} J$,
then we set $\dff^P(I,J)=\infty$. It is clear that $\dff^P(I,J)$ is
well-defined. 
Moreover, as $I \approx^{P_{<1}} J$, $\dff^P(I,J)\geq
1$. Furthermore, for a set $V\subseteq 2^\U$ and a relation $\succ$
over $2^\U$, we write $\succ_{V}$ for the restriction of $\succ$ to
$V$, that is, $\succ_{V}\ = \{(A,B)\in\,\, \succ \mid A,B\in V\}$.

\begin{theorem}\label{thm:sequivrankedprefeq}
For every ranked optimization problems $P$ and $Q$, and every rank
interval $[i,j]$, $P \equiv^{s,[i,j]} Q$ if and only if the following
conditions hold:
\begin{enumerate}
\item $\pref(P_{<i})=\pref(Q_{<i})$
\item ${>^P_{\pref(P_{<i})}}={>^Q_{\pref(Q_{<i})}}$
\item For every $I,J\in \pref(P_{<i})$ such that $i<\dff^P(I,J)$ or
$i<\dff^Q(I,J)$, $\dff^P(I,J)=\dff^Q(I,J)$ or both $\dff^P(I,J)>j$ and
$\dff^Q(I,J)>j$.
\end{enumerate}
\end{theorem}

We now comment on this characterization and derive some of its consequences.
First,
we observe that the conditions (1) and (2) are indeed necessary --- 
differences 
between $\pi(P_{<i})$ and $\pi(Q_{<i})$ or ${>^P_{\pref(P_{<i})}}$ and
${>^Q_{\pref(Q_{<i})}}$ can be exploited to construct a selector
from $\L^{s,[i,j]}$ whose addition to $P$ and $Q$ results in problems
with different sets of optimal outcomes. This is illustrated by
Examples \ref{ex:3n} and \ref{ex:4n} in the case of simple problems and
simple contexts ($i=j=1$), where $\pi(P_{<i})$ and $\pi(Q_{<i})$ coincide
with 
$\mu(P)$ and $\mu(Q)$, %
respectively. The condition (3) is
necessary, too. Intuitively, if the first ranks where $P$ and $Q$ 
differentiate between two outcomes $I$ and $J$ (which are optimal for
ranks less than $i$) are not equal, these first ranks must both be 
larger than $j$. Otherwise, one can find a selector with rules of ranks 
in $[i,j]$, that will make one of the interpretation optimal in one 
extended problem but not in the other.  

Next, we discuss some special cases of the characterization. First, we
consider the case $i=1$, which allows for a simplification of
Theorem~\ref{thm:sequivrankedprefeq}.
\begin{corollary}\label{cor:sequivrankedprefeq1}
For every ranked optimization problems $P$ and $Q$, and every rank
interval $[1,j]$, $P \equiv^{s,[1,j]} Q$ if and only if the following
conditions hold:
\begin{enumerate}
\item $\mu(P)=\mu(Q)$
\item ${>^P_{\mu(P)}}={>^Q_{\mu(Q)}}$
\item For every $I,J\in \mu(P)$, $\dff^P(I,J)=\dff^Q(I,J)$ or both $\dff^P(I,J)>j$ and $\dff^Q(I,J)>j$.
\end{enumerate}
\end{corollary}
\begin{proof}
Starting from 
Theorem~\ref{thm:sequivrankedprefeq}, we
note that the se\-lector of $P_{<1}$ is empty and hence $\pi(P_{<1})
=\mu(P)$. Moreover, if the precondition $i<\dff^P(I,J)$ and
$i<\dff^Q(I,J)$ in condition (3) of Theorem~\ref{thm:sequivrankedprefeq}
is not satisfied for $i=1$ and a pair $I,J\in \mu(P)$, then
$\dff^P(I,J)=1$ and $\dff^Q(I,J)=1$ and thus the consequent is
satisfied in that case as well, which allows for omitting the
precondition.
\end{proof}

If in addition $j = \infty$, we obtain the case of rank-unrestricted
selector contexts, and condition (3) 
can be simplified once more, since $\dff^P(I,J)>j$ and
$\dff^Q(I,J)>j$ never hold for $j = \infty$.
\begin{corollary}\label{cor:sequivrankedpref}
For every optimization problems $P$ and $Q$, $P \equiv^{s} Q$ (equivalently, $P \equiv^{s,\geq 1} Q$ or $P \equiv^{s,[1,\infty]} Q$) if and 
only if the following conditions hold:
\begin{enumerate}
\item $\mu(P)=\mu(Q)$
\item ${>^P_{\mu(P)}}={>^Q_{\mu(Q)}}$
\item for every $I,J\in \mu(P)$, $\dff^P(I,J) =\dff^Q(I,J)$.
\end{enumerate}
\end{corollary}

Next, we note that if an optimization problem is simple then
$\dff^P(I,J)>1$ if and only if $\dff^P(I,J)=\infty$, which is
equivalent to $I\approx^P J$. This observation leads to the following
characterization of strong sel-equivalence of simple optimization
problems.
\begin{corollary}\label{cor:simples}
For every two \emph{simple} optimization problems $P$ and $Q$, the following 
statements are equivalent:
\begin{enumerate}
\item[(a)] $P \equiv^s Q$ (equivalently, $P \equiv^{s,[1,\infty]} Q$)
\item[(b)] $P \equiv^{s,=1} Q$ (equivalently, $P \equiv^{s,[1,1]} Q$)
\item[(c)] $\mu(P)=\mu(Q)$ and $\geq^P_{\mu(P)} = \geq^Q_{\mu(Q)}$.
\end{enumerate}
\end{corollary}
\begin{proof}
The implication (a)$\Rightarrow$(b) is evident from the definitions.

\smallskip
\noindent
(b)$\Rightarrow$(c)
From Corollary~\ref{cor:sequivrankedprefeq1} with $j=1$ we directly
obtain $\mu(P)=\mu(Q)$. %
The condition $\geq^P_{\mu(P)}=\geq^Q_{\mu(Q)}$ 
follows from conditions (2) and (3) of that corollary. Indeed, let us
consider $I,J\in \mu(P)$ such that $I \geq^P J$ and distinguish two 
cases. If (i) $\dff^P(I,J)=1$ then $I >^P J$ and
by condition (2) of Corollary~\ref{cor:sequivrankedprefeq1}, also 
$I >^Q J$, implying $I \geq^Q J$. If (ii) $\dff^P(I,J)>1$ then by 
condition (3) of Corollary~\ref{cor:sequivrankedprefeq1}, $\dff^Q(I,J)>1$.
Since $P, Q$ are simple, $I \approx^Q J$, and consequently
$I \geq^Q J$. By symmetry, we also have that $I \geq^Q J$ implies 
$I \geq^P J$. Thus, $\geq^P_{\mu(P)}=\geq^Q_{\mu(Q)}$.

\smallskip
\noindent
(c)$\Rightarrow$(a)
From (c) it follows that $>^P_{\mu(P)} = >^Q_{\mu(Q)}$ and $\approx^P_{\mu(P)} = \approx^Q_{\mu(Q)}$.
Thus, the conditions (1) and (2) of Corollary~\ref{cor:sequivrankedpref} follow. To prove the condition (3), 
let us first
assume $\dff^P(I,J)> 1$ for $I,J \in \mu(P)$. It follows that $\dff^P(I,J)=\infty$ and thus
$I\approx^P_{\mu(P)} J$. By our earlier observation also $I\approx^Q_{\mu(Q)} J$ and thus $\dff^Q(I,J)=\infty$. Hence $\dff^P(I,J)=\dff^Q(I,J)$. For $\dff^Q(I,J)> 1$ we reason analogously. In the last remaining case, $\dff^P(I,J)=1$ and 
$\dff^Q(I,J)=1$, so we directly obtain $\dff^P(I,J)=\dff^Q(I,J)$. By 
Corollary~\ref{cor:sequivrankedpref}, $P \equiv^{s} Q$ follows.
\end{proof}

Corollary \ref{cor:simples} shows, in particular, that for simple 
problems there is no difference between the relations $\equiv^{s,\geq 1}$ 
and $\equiv^{s,=1}$. This property reflects the role of preference rules
of rank 2 and higher. They allow us to break ties among optimal outcomes, 
as defined by preference rules of rank 1. Thus, they can eliminate some 
of these outcomes from the family of optimal ones, but they cannot 
introduce new optimal outcomes. Therefore, they do not affect strong
sel-equivalence of simple problems. This property has the following 
generalization to ranked optimization problems.

\begin{corollary}
\label{cor:bigranks}
Let $P$ and $Q$ be ranked optimization problems and let $k$ be the
maximum rank of a preference rule in $P\cup Q$. Then the relations
$\equiv^{s,\geq k}$ and $\equiv^{s,=k}$ coincide.
\end{corollary}
\begin{proof}
Clearly, $P\equiv^{s,\geq k}Q$ implies $P\equiv^{s,=k} Q$. Thus,
it is enough to prove that if $P\equiv^{s,=k} Q$ then $P\equiv^{s,\geq k}Q$.
To prove the condition 
(3), let us consider $I,J\in\pref(P_{<k})$ such that $\dff^P(I,J)>k$. By 
the condition (3) of Theorem \ref{thm:sequivrankedprefeq}, $\dff^Q(I,J)>
k$. Since $k$ is the maximum rank of a preference rule in $P$ or $Q$, 
$\dff^P(I,J)=\infty$ and $\dff^Q(I,J)=\infty$. Thus, $\dff^P(I,J)=
\dff^Q(I,J)$ (the case $\dff^Q(I,J)>k$ is similar).
\end{proof}

Our observation on the role of preference rules with ranks higher than 
ranks of rules in $P$ or $Q$ also implies that $P$ and $Q$ are strongly
sel-equivalent relative to selectors consisting exclusively of such rules if
and only if $P$ and $Q$ are equivalent (have the same optimal outcomes), 
and if optimal outcomes that ``tie'' in $P$ also ``tie'' in $Q$ and
conversely. Formally, we have the following result.

\begin{corollary}
\label{cor:biggerranks}
Let $P$ and $Q$ be ranked optimization problems and let $k$ be the
maximum rank of a preference rule in $P\cup Q$. Then $P\equiv^{s,\geq k+1}
Q$ if and only if $\pi(P)=\pi(Q)$ and $\approx^P_{\pref(P)}=
\approx^Q_{\pref(Q)}$.
\end{corollary}
\begin{proof}
Clearly, $P_{<k+1}=P$ and $Q_{<k+1}=Q$ and so, $\pref(P_{<k+1})=\pref(P)$
and $\pref(Q_{<k+1})=\pref(Q)$. Thus, the ``only-if'' part follows by 
Theorem \ref{thm:sequivrankedprefeq} (the condition (1) of that theorem
reduces to $\pi(P)=\pi(Q)$ and the condition (3) implies $\approx^P_{\pref(P)}=
\approx^Q_{\pref(Q)}$). To prove the ``if'' part, we note that
the condition (1) of Theorem \ref{thm:sequivrankedprefeq} holds by the 
assumption. Moreover, the relations $>^P_{\pref(P)}$ and $>^Q_{\pref(Q)}$ 
are empty and so, they coincide. Thus, the condition (2) of Theorem 
\ref{thm:sequivrankedprefeq} holds. Finally, if $I,J\in\pref(P)$, and 
$\dff^P(I,J)>k+1$, then $\dff^P(I,J)=\infty$ and so, $I\approx^P J$. By
the assumption, $I\approx^Q J$, that is, $\dff^Q(I,J)=\infty=\dff^P(I,J)$.
The case when $\dff^Q(I,J)>k+1$ is similar.
Thus, the condition~(3) of Theorem \ref{thm:sequivrankedprefeq} holds, too, and
$P\equiv^{s,\geq k+1} Q$ follows.
\end{proof}

Lastly, we give some simple examples illustrating how our results 
can be used to ``safely'' modify or simplify optimization problems, 
that is rewrite one into another strongly sel-equivalent one.

\begin{example}
\label{ex:5n}
Let $P=(T,S)$, where $T=\{a\vee b\vee c, \neg(a\land b),
\neg(a\land c), \neg(b\land c)\}$ and $S=\{a>c\LPif,\ b>c\LPif\}$, 
and $P'=(T,S')$, where $S'=\{a\vee b > c\LPif\}$. Regarding these
problems as CO problems, we have that $\mu(P)=\mu(P')=\{\{a\},\{b\},
\{c\}\}$. Moreover, it is evident that $\geq^P_{\mu(P)}=\geq^Q_{\mu(Q)}$. 
Thus, by Corollary \ref{cor:simples}, $P$ and $P'$ are strongly sel-equivalent.
In other words, we can faithfully replace rules 
$a>c\LPif$, $b>c\LPif$ in the selector of any 
optimization problem 
with generator $T$
by the single rule $a\vee b > c\LPif$.
\end{example}

Next, for an example of a more general principle, we note that removing 
preference rules with only one option in the head yields a problem that 
is strongly sel-equivalent.

\begin{corollary}\label{cor:rem}
Let $P$ and $Q$ be two CO or ASO problems such that $P^g=Q^g$ and
$Q^s$ is obtained from $P^s$ by removing all preference rules with only
one option in the head. Then $P$ and $Q$ are strongly sel-equivalent.
\end{corollary}
\begin{proof}
The conditions (1)-(3) of Theorem \ref{thm:sequivrankedprefeq} all
follow from an observation that for every interpretation $I$ and every 
preference rule $r$ with just one option in the head, $v_I(r)=1$.
\end{proof}

\section{Strong gen-equivalence}
\label{gen-eq}

We 
now
focus on the case of strong gen-equivalence. The 
semantics of generators makes a difference here but the difference
concerns only the fact that under the two semantics we consider,
the concepts of strong equivalence are different. Other aspects of 
the characterizations are the same. Specifically, generators have to 
be strongly 
equivalent relative to a selected semantics. Indeed, if the generators 
are not strongly equivalent, one can extend them uniformly so that after 
the extension one problem has a single outcome, which is then trivially 
an optimal one, too, while the other one has no outcomes and so, no 
optimal ones. Second, the preference relation $>$ defined by the 
selectors of the problems considered must coincide. Thus, a single 
theorem handles both types of problems. 

\begin{theorem}
\label{thm:simpleg}
For every two CO (ASO, respectively) problems $P$ and $Q$, $P \equiv_g Q$ 
if and only if $P^g$ and $Q^g$ are strongly equivalent (that is, 
$\Mod{P^g}=\Mod{Q^g}$ for CO problems, and $\HT{P^g} = \HT{Q^g}$
for ASO problems) and ${>^P_{{\Mod{P^g}}}}= {>^Q_{{\Mod{Q^g}}}}$.
\end{theorem}

In view of Examples \ref{ex:1n} and \ref{ex:2n}, the result is not unexpected.
The two examples demonstrated that the conditions of the characterization
cannot, in general, be weakened.

It is clear from Corollary \ref{cor:sequivrankedpref} and Theorem 
\ref{thm:simpleg} that strong sel-equivalence of CO problems is a 
stronger property than their strong gen-equivalence. 

\begin{corollary}\label{cor:coincide}
For every two CO problems $P$ and $Q$, $P \equiv^s Q$ implies
$P \equiv_g Q$.
\end{corollary}

In general the implication in Corollary \ref{cor:coincide} cannot be
reversed. The problems $P_6$ and $P_7$ considered in Example \ref{ex:4n}
are not strongly sel-equivalent. However, based on Theorem~\ref{thm:simpleg}, 
they are strongly gen-equivalent. Indeed, $\Mod{P_6^g}=\Mod{P_7^g}$ and, 
writing $M$ for $\Mod{P_6^g}=\Mod{P_7^g}$, the relations $>^{P_6}_M$ and 
$>^{P_7}_M$ are both empty and so, equal. 

The relation between strong sel-equivalence and strong gen-equivalence
of ASO problem is more complex. In general, neither property implies 
the other even if both problems $P$ and $Q$ are assumed to be simple.
It is so because $P\equiv^s Q$ if and only if $AS(P^g)=AS(Q^g)$ and 
$\geq^P_{AS(P^g)} = \geq^P_{AS(P^g)}$ (Corollary \ref{cor:simples}), 
and $P\equiv^g Q$ if and only of $\HT{P^g}=\HT{Q^g}$ and $>^P_{\Mod{P^g}} 
= >^P_{\Mod{Q^g}}$ (Theorem \ref{thm:simpleg}). Now, $AS(P^g)=AS(Q^g)$ 
(regular equivalence of programs) does not imply $\HT{P^g}=\HT{Q^g}$ 
(strong equivalence) and $>^P_{\Mod{P^g}} = >^P_{\Mod{Q^g}}$ does not imply 
$\geq^P_{AS(P^g)} = \geq^P_{AS(P^g)}$.

\section{Strong equivalence --- the combined case} 
\label{eq-comb}

Finally, we consider the relation $\equiv^s_g$, which results from
considering contexts that combine both generators and selectors. Since
generators may vary here, as in the previous section, the semantics of
generators matters. But, as in the previous section, the difference
boils down to different characterizations of strong equivalence of 
generators.

We start with a result characterizing strong equivalence of CO and 
ASO problems relative to combined contexts (both generators and selectors 
possibly non-empty) with selectors consisting of rules of rank at least $i$ and at 
most $j$, respectively. 

\begin{theorem}\label{thm:cequivrankedprefeqcomb}
For every ranked CO (ASO, respectively) problems $P$ and $Q$,
and every rank interval $[i,j]$,
$P \equiv^{s,[i,j]}_g Q$ if and only if the following conditions hold:
\begin{enumerate}
\item $P^g$ and $Q^g$ are strongly equivalent (that is, $\Mod{P^g}=
\Mod{Q^g}$ for CO problems, and $\HT{P^g}=\HT{Q^g}$ for ASO problems)
\item ${>^P_{\Mod{P^g}}}={>^Q_{\Mod{Q^g}}}$
\item For every $I,J\in \Md(P^g)$ such that $i<\dff^P(I,J)$ or
$i<\dff^Q(I,J)$, $\dff^P(I,J)=\dff^Q(I,J)$ or both $\dff^P(I,J)>j$ and
$\dff^Q(I,J)>j$
\item ${>^{P_{<i}}_{\Mod{P^g}}}= {>^{Q_{<i}}_{\Mod{Q^g}}}$.
\end{enumerate}
\end{theorem}

The corresponding characterizations for CO and ASO problems differ 
only in their 
respective conditions (1), which now reflect different conditions 
guaranteeing strong equivalence of generators under the classical and 
answer-set 
semantics. Moreover, the four conditions of Theorem 
\ref{thm:cequivrankedprefeqcomb} can be obtained by suitably 
combining and extending the conditions of Theorem 
\ref{thm:sequivrankedprefeq} and Theorem \ref{thm:simpleg}. First, 
as combined strong equivalence implies strong gen-equivalence, the 
condition (1) is taken from Theorem \ref{thm:simpleg}. 
Second, we modify the conditions (2) and (3) from Theorem 
\ref{thm:sequivrankedprefeq} replacing $\pi(P_{<i})$ with $\Md(P^g)$
(and accordingly  $\pi(Q_{<i})$ with $\Md(Q^g)$), as 
each classical model of $P^g$ can give rise to an optimal classical or 
equilibrium one upon the addition of a context, an aspect also
already visible in Theorem \ref{thm:simpleg}. Finally, we have to add 
a new condition that the relations $>^{P_{<i}}$ and $>^{Q_{<i}}$ coincide 
on the sets of models of $P^g$ and $Q^g$. When generators are allowed
to be extended, one can make any two of their models to be the only
outcomes after the extension. If the two outcomes, say $I$ and $J$,
are related differently by the corresponding strict relations induced 
by rules with ranks less than $i$, then in one extended problem, exactly 
one of the two outcomes, say $I$, is optimal. In the other extended
problem we cannot have both outcomes be optimal nor $J$ only be optimal 
as that
would contradict the strong equivalence of problems under considerations.
If, however, $I$ is the only optimal outcome also in the other extended
problem, then $I$ must ``win'' with $J$ based on rules that have ranks
at least $j$. In such case, there is a way to add new preferences of
rank $i$ that will ``promote'' $J$ to be optimal too, without making
it optimal in the first problem. However, that contradicts strong
equivalence. 

We conclude this section with observations concerning the relation
$\equiv^{s}_g$ for both CO and ASO problems. The contexts relevant 
here may contain preference rules of arbitrary ranks. We start with the
case of CO problems, where the results are stronger. While they can be 
derived from the general theorems above, we will present here arguments
relying on results from previous sections, which is possible since for
CO problems equivalence and strong-equivalence of generators coincide.

We saw in the last section that for CO problems $\equiv^s$ is a strictly 
stronger relation than $\equiv^g$. In fact, for CO problems, $\equiv^s$ 
coincides with the general relation $\equiv^s_g$.

\begin{theorem}\label{thm:coincide}
For every CO problems $P$ and $Q$, $P \equiv^s_g Q$ if and only if
$P \equiv^s Q$.
\end{theorem}
\begin{proof}
The ``only-if'' implication is evident. To prove the converse implication,
we will use Proposition \ref{prop:ortho} which reduces checking for
strong equivalence to checking for strong sel-equivalence. Let $R\in\L^g$
be a generator problem. Since $P \equiv^s Q$, from Corollary
\ref{cor:sequivrankedpref} we have $\Md(P^g)=\Md(Q^g)$. Consequently,
$\Md((P\cup R)^g)=\Md((Q\cup R)^g)$. Writing $M$ for $\Md(P^g)$ and $M'$
for $\Md((P\cup R)^g)$ we have $M'\subseteq M$. Thus, also by Corollary
\ref{cor:sequivrankedpref}, $>^{P\cup R}_{M'}=>^{Q\cup R}_{M'}$. 
Finally, the condition (3) of Corollary \ref{cor:sequivrankedpref} for
$P$ and $Q$ implies the condition (3) of that corollary for $P\cup R$
and $Q\cup R$ (as $R$ has no preference rules and $M'\subseteq M)$. 
It follows, again by Corollary \ref{cor:sequivrankedpref}, that 
$P\cup R \equiv^s Q\cup R$. Thus, by Proposition~\ref{prop:ortho},
$P \equiv^s_g Q$.
\end{proof}

In particular, Corollary \ref{cor:simples} implies that the relations 
$\equiv^s_g$, $\equiv^{s,=1}_g$, $P\equiv^{s,=1} Q$, and $\equiv^s$ 
coincide on simple CO problems.

\begin{corollary}\label{cor1:coincide}
For every simple CO problems $P$ and $Q$ all properties $P \equiv^s_g Q$, 
$P \equiv^{s,=1}_g Q$, $P \equiv^{s,=1} Q$ and $P \equiv^s Q$ are 
equivalent.
\end{corollary}

For simple ASO problems we still have that $\equiv^{s}_g$ and $\equiv^{s,=1}_g$
coincide but in general these notions are different from $\equiv^s$ and
$\equiv^{s,=1}$.

\begin{corollary}\label{cor:char-aso-combined}
For every simple ASO problems $P$ and $Q$, the following conditions are
equivalent
\begin{enumerate}
\item[(a)] $P \equiv_g^s Q$
\item[(b)] $P \equiv_g^{s,=1} Q$
\item[(c)] $\HT{P^g}{=}\HT{Q^g}$ and ${\geq^P_{\Md(P^g)}} = {\geq^Q_{\Md(Q^g)}}$.
\end{enumerate}
\end{corollary}
\begin{proof}
The implication (a)$\Rightarrow$(b) is evident. 

Let us assume (b). By
Theorem \ref{thm:cequivrankedprefeqcomb}, we have $\HT{P^g}{=}\HT{Q^g}$.
This identity implies $\Md(P^g)=\Md(Q^g)$. Let us assume that for some 
$I,J\in\Md(P^g)$, $I\geq^P_{\Md(P^g)} J$. If $I>^P_{\Md(P^g)} J$ then,
by Theorem \ref{thm:cequivrankedprefeqcomb}, $I>^Q_{\Md(Q^g)} J$ and
so, $I\geq^Q_{\Md(Q^g)} J$. Otherwise, $I\approx^P J$ and so, $\dff^P(I,J)
=\infty$. By Theorem~\ref{thm:cequivrankedprefeqcomb}, $\dff^Q(I,J)>1$.
Since $Q$ is simple, $\dff^Q(I,J)=\infty$. Thus, $I\approx^Q J$ and, also,
$I\geq^Q_{\Md(Q^g)} J$. The converse implication follows by symmetry. Thus,
(c) holds.   

Finally, we assume (c) and prove (a). To this end, we show that the
conditions (1)--(4) of Theorem \ref{thm:cequivrankedprefeqcomb} hold. 
Directly from the assumptions, we have that condition~(1) holds. 
Condition (2) follows from the general fact that
for every optimization problems $P$ and $Q$, and every set $V\subseteq 2^\U$,
${\geq^P_V} = {\geq^Q_V}$ implies
${>^P_V} = {>^Q_V}$.
Moreover, we also have that $\Md(P^g)=\Md(Q^g)$. To prove the 
condition (3), let us assume that $I,J\in\Md(P^g)$ and that 
$\dff^P(I,J)>1$. Since $P$ is simple, $I\approx^P J$. Thus, $I\approx^Q J$
and, consequently, $\dff^P(I,J)=\infty=\dff^Q(I,J)$. Finally, 
condition (4), i.e.\
${>^{P_{<i}}_{\Mod{P^g}}}= {>^{Q_{<i}}_{\Mod{Q^g}}}$,
obviously holds in case $i=1$ and $\Md(P^g)=\Md(Q^g)$.
\end{proof}

\section{Complexity}
\label{sec:complexity}

In this section, we study the problems of deciding the various notions
of strong equivalence. Typically the comparisons between sets of outcomes in the
characterizations determine the respective complexity. We start with
results concerning strong sel-equivalence.
\begin{theorem}\label{thm:co-sel:simple}
Given optimization problems $P$ and $Q$, 
deciding $P\equiv^{s} Q$ is $\CONP$-complete in case of CO-problems 
and $\PiP{2}$-complete in case of ASO-problems.
\end{theorem}
\begin{proof}[Sketch]
For membership, one can show that given a pair of interpretations $I,J$
it can be verified in polynomial time (for CO-problems) or in
polynomial time using an $\NP$ oracle (for ASO-problems) whether they
form a witness for the complement of the conditions stated in
Corollary~\ref{cor:sequivrankedpref}. The main observation is that
model checking is polynomial for the classical semantics, but
$\CONP$-complete for the equilibrium semantics (Theorem~8 of
\cite{PearceTW09}).

Hardness follows from considering the equivalence problem for
optimizations problems with empty selectors, which is known to be
$\CONP$-hard (for classical semantics) and $\PiP{2}$-hard (for
equilibrium semantics, Theorem~11 of \cite{PearceTW09}).
\end{proof}

For the ranked case, we observe an increase in complexity, which can
be explained by the characterization given in
Theorem~\ref{thm:sequivrankedprefeq}: Instead of outcome checking,
this characterization involves optimal outcome checking, which is more
difficult (unless the polynomial hierarchy collapses).
\begin{theorem}\label{thm:co-sel:ranked}
Given optimization problems $P$ and $Q$ and a rank interval $[i,j]$,
deciding $P\equiv^{s,[i,j]} Q$ is $\PiP{2}$-complete in case of CO-problems 
and $\PiP{3}$-complete in case of ASO-problems.
\end{theorem}
\begin{proof}[Sketch]
The membership part essentially follows the same arguments as the
proof of Theorem~\ref{thm:co-sel:simple}, but here the
problem of checking $I\in\pref(P_{<i})$ is in $\CONP$ for CO-problems
and in $\PiP{2}$ for ASO-problems. 

For the hardness part, we reduce the following problem to sel-equivalence of CO-problems:
Given two propositional theories $S$ and $T$, decide whether they possess
the same minimal models. This problem is known to 
be $\PiP{2}$-complete 
(e.g.\ Theorem~6.15 \cite{efw04}), and 
the problem remains hard if $S$ and $T$ are in NNF 
given over the
same alphabet.
We adapt a construction used in \cite{bnt-unpublished}, and
given a theory $T$ (over atoms in $U$) we construct a CO problem 
$P_T$ where
\begin{eqnarray*}
P^g_T &=& T[\neg u/u'] \cup \{  u \leftrightarrow \neg u' \mid u \in U\},\\
P^s_T &=& \{ u' > u \leftarrow \mid u \in U\},
\end{eqnarray*}
{and $T[\neg u/u']$ stands for the theory resulting from
replacing all $\neg u$ by $u'$ in $T$.}
The elements in $\pi(P_T)$
are in a one-to-one correspondence to the minimal models of $T$.
For theories $S$ and $T$ over $U$ it follows 
that 
$S$ and $T$ have the same minimal models if and only if
$P_S \equiv^{s,\geq 2} P_T$.

Concerning the hardness part for ASO problems, we use the following 
problem: given two open QBFs
$\forall Y \phi(X,Y)$,
$\forall Y \psi(X,Y)$, do they possess 
the same minimal models. This problem is
$\PiP{3}$-hard\LongOnly{ (see Lemma~\ref{lem:compl-qbf-min-mod} in Appendix A)}.
For $\phi(X,Y)$, 
we construct $P_\phi$  as follows:
\begin{eqnarray*}
P^g_\phi &=& \{ z \vee z' \mid z\in X\cup Y\}\cup\\
        && \{ (y\wedge y')\rightarrow w, w\rightarrow y, w\rightarrow y'  \mid y\in Y\}\cup\\
        && \{ \phi[\neg z/z'] \rightarrow w,
\neg w \rightarrow w \},\\
P^s_\phi &=& \{ x' > x \leftarrow \mid x \in X\},
\end{eqnarray*}
{where $\phi[\neg z/z']$ stands for the formula obtained by
replacing all $\neg z$ by $z'$ in $\phi(X,Y)$.}
The elements in $\pi(P_\phi)$
are in a one-to-one correspondence to the minimal models of $\forall Y\phi(X,Y)$.
For $\phi$ and $\psi$ over $X\cup Y$ we get
that 
$\forall Y\phi(X,Y)$ and $\forall Y\psi(X,Y)$
have the same minimal models if and only if
$P_\phi \equiv^{s,\geq 2} P_\psi$.
\end{proof}

In Theorem~\ref{thm:co-sel:ranked} the rank interval $[i,j]$
is given in input. When fixing the interval, the hardness results
still hold, provided that $i>1$. In fact, the critical condition in
Corollary~\ref{cor:sequivrankedprefeq1} is
$\pref(P_{<i})=\pref(Q_{<i})$; for rank intervals $[1,j]$, the
selectors become empty and the condition is reduced to
$\mu(P)=\mu(Q)$, which is easier to decide.

The remaining problems are all in $\CONP$. For strong gen-equivalence,
completeness follows directly from Theorem~\ref{thm:simpleg} and
$\CONP$-completeness of deciding strong equivalence between two
propositional theories (for both semantics).
\begin{theorem}\label{thm:co-gen}
Given two CO (ASO, respectively) problems $P$ and $Q$, deciding $P \equiv_g Q$
is $\CONP$-complete.
\end{theorem}

Finally, for the combined case the hardness result follows from
Theorem~\ref{thm:cequivrankedprefeqcomb} and $\CONP$-completeness of
deciding strong equivalence of propositional theories.
\begin{theorem}\label{thm:co-combined}
Given ranked CO (ASO, respectively) problems $P$ and $Q$,
and rank interval $[i,j]$,
deciding $P \equiv^{s,[i,j]}_g Q$ is $\CONP$-complete. 
\end{theorem}

By construction, all hardness results hold already for
simple optimization problem.

\section{Discussion}

We introduced the formalism of optimization problems, generalizing
the principles of ASO programs, in particular, the separation of hard 
and soft constraints~\cite{bnt03}. We focused on two important 
specializations of optimization problems: CO problems and ASO problems.
We studied various forms of strong equivalence for these classes of 
optimization problems, depending on what contexts are considered.
Specifically, we considered the following cases: new preference 
information is added, but the hard constraints remain unchanged (strong 
sel-equivalence); hard constraints are added but preferences remain 
unchanged (strong gen-equivalence); both hard constraints and 
preferences can be added (strong equivalence). 
To the best of our knowledge, this natural classification of equivalences
in preference formalisms has not been studied yet.
In certain cases some of these notions coincide (Theorem~\ref{thm:coincide})
but this is no longer true when the underlying semantics is changed
or ranks in contexts are restricted.

In fact, we established characterizations of all these notions of strong 
equivalence. They exhibit strong similarities. The characterizations 
of strong sel-equivalence for CO and ASO problems in Theorem 
\ref{thm:sequivrankedprefeq}
are precisely the same, mirroring 
the fact that generators are not subject to change. Theorem 
\ref{thm:simpleg} concerns strong gen-equivalence for CO and ASO problems. 
In each case, the characterizations consist of two requirements: the
strong equivalence of generators, and the equality of the strict 
preference relations restricted to the class of models of the 
generators. The only difference comes from the fact that strong 
equivalence for classical and the equilibrium-model semantics have 
different characterizations. Theorem \ref{thm:cequivrankedprefeqcomb} 
which concerns the combined case of strong equivalence 
also does not differentiate between CO and
ASO problems other than implicitly (as before, the conditions of strong
equivalence are different for the two semantics). Moreover, the
characterizations given arise in a certain systematic way from 
those given in Theorems \ref{thm:sequivrankedprefeq} and~%
\ref{thm:simpleg}. This being the case in each of the different semantics
we used strongly suggests that there are some abstract principles at
play here. We are currently pursuing this direction, conjecturing that 
this is an inherent feature of preference formalisms with separation of 
hard and soft constraints.

Our results give rise to problem rewriting methods that transform
optimization problems into strongly equivalent ones. We provided two
simple examples illustrating that application of our results 
in Example~\ref{ex:5n} and Corollary~\ref{cor:rem}.
Similar examples can be constructed
for our results concerning strong gen-equivalence and (combined) strong 
equivalence. A more systematic study of optimization problem rewriting
rules that result in strongly equivalent problems will be a subject of
future work.

Finally, we established the complexity of deciding whether optimization
problems are strongly equivalent. Notably, in the general case of 
strong (combined) equivalence this problem remains in $\CONP$ 
for both CO and ASO problem.
It is 
strong sel-equivalence that is computationally hardest to test 
(in case of ASO problems, $\PiP{3}$-hard).
It is so, 
as the concept depends of properties of outcomes that are optimal with 
respect to rules of ranks less than $i$, while in other cases all models
have to be considered. Testing optimality is harder than testing
for being a model, explaining the results we obtained.  

\section{Acknowledgments}
The first author was supported by Regione Calabria and EU under POR
Calabria FESR 2007-2013 within the PIA project of DLVSYSTEM s.r.l.,
and by MIUR under the PRIN project LoDeN. The second author was
supported by the NSF grant IIS-0913459.

\bibliography{add}
\bibliographystyle{aaai}

\LongOnly{

\section*{Appendix A: Useful Lemmas}

We provide here several lemmas that we use later in the proofs of the
results discussed in the main body of the paper.
The first %
property follows immediately from the definitions.

\begin{lemma}\label{lemma:simple}
Let $P$ and $Q$ be optimization problems with %
${>^P_{\acc(P)}} = {>^Q_{\acc(Q)}}$. 
Then, $\pref(P)=\pref(Q)$.
\end{lemma}

The next lemma characterizes the relation $>^{P\cup Q}$ for %
ranked optimization problems $P$ and $Q$.
We recall that %
$\dff^P(I,J)$ is the largest $k$ such that 
$I \approx^{P_{<k}} J$, where in addition $\dff^P(I,J)=\infty$, if $I 
\approx^{P_{<k}} J$, for every $k$.

\begin{lemma}\label{lemma:compose}
Let $P$ and $Q$ be optimization problems, and $I,J$ be interpretations.
Then, $I>^{P\cup Q} J$ holds if and only if one of the following conditions
holds:
\begin{enumerate}
\item $\dff^P(I,J)<\dff^Q(I,J)$ and $I>^P J$;
\item $\dff^P(I,J)>\dff^Q(I,J)$ and $I>^Q J$;
\item $\dff^P(I,J)=\dff^Q(I,J)$, $I>^P J$ and $I>^Q J$. 
\end{enumerate}
\end{lemma}
\begin{proof}
The ``if'' direction is evident. To prove the ``only-if'' direction,
we note that the cases $\dff^P(I,J)<\dff^Q(I,J)$ and $\dff^P(I,J)>
\dff^Q(I,J)$ are obvious, too. Thus, let us assume $\dff^P(I,J)=\dff^Q(I,J)=i$.
Clearly, $i<\infty$ (otherwise, $I\approx^{P\cup Q} J$, contrary to the
assumption). It follows that for every rule $r\in P^s\cup Q^s$ of rank
less than $i$, $v_I(r)=v_J(r)$. Next, for every $r\in P^s\cup Q^s$ 
of rank $i$, $v_I(r)\geq v_J(r)$. Finally, there are rules $r\in P^s$ and
$r'\in Q^s$, each of rank $i$ such that $v_I(r)\not= v_J(r)$ and
$v_I(r')\not= v_J(r')$ (since $\dff^P(I,J)=i$ and $\dff^Q(I,J)=i$). It
follows that $v_I(r)<v_J(r)$ and $v_I(r')<v_J(r')$. Thus,
$I>^P J$ and $I>^Q J$, as needed.
\end{proof}

\begin{lemma}\label{lemma:lowerrank}
For every optimization problem $P$ and every $i \geq 1$,
$\pref(P_{<i})\supseteq\pref(P)$.
\end{lemma}
\begin{proof}
Let us assume that $I\notin \pref(P_{<i})$. Then, there is an
interpretation $J\in \acc(P_{<i})$ such that $J>^{P_{<i}} I$.
Thus, there is a rule $r\in P^s_{<i}$, say of rank $j$, such that 
(i) $v_J(r)<v_I(r)$;
(ii) for every $r'\in P^s_{<i}$ with the rank $j$, $v_J(r')\leq v_J(r')$;
and
(iii) for every $r'\in P^s_{<i}$ with rank less then $j$, $v_J(r')= v_I(r')$.
We note that, since $\acc(P_{<i})=\acc(P)$, $J\in \acc(P)$. Moreover,
$j<i$ and so the sets of rules with ranks less than or equal to $j$ in 
$P^s$ and $P^s_{<i}$ coincide. Thus, $J>^P I$ follows and, consequently,
$I\notin\pref(P)$.
\end{proof}

We give the next two lemmas without proofs, as they are easy
consequences of results by Ferraris [\citeyear{fer05}] and Ferraris
and Lifschitz [\citeyear{fl05}].

\begin{lemma}
\label{lemma:as-elim1}
Let $P$ be a theory, $I$ an interpretation, and let
$\Pi[I]=\{a\rightarrow \bot\st a\notin I\}
\cup \{\neg a\rightarrow \bot\st a\in I\}$. Then, $\AS{P\cup \Pi[I]}=
\Mod{P\cup \Pi[I]}=\{I\}$.
\end{lemma}

\begin{lemma}
\label{lemma:as-elim2}
Let $P$ be a theory, $I,J$ two of its (classical) models such that $I\not=J$,
and let
\begin{eqnarray*}
\Pi[I,J] &=&\{ a\vee b \st a\in I,\ b\in J\}\\
 &\cup&\{a\vee\neg b \st a\in I,\ b\notin J\}\\
 &\cup&\{\neg a\vee b \st a\notin I,\ b\in J\}\\
 &\cup&\{\neg a\vee\neg b \st a\notin I,\ b\notin J\}.
\end{eqnarray*}
Then $\AS{P\cup \Pi[I,J]}=\Mod{P\cup \Pi[I,J])}=\{I,J\}$.
\end{lemma}

\begin{lemma}\label{lemma:selnew}
Let $P$ be an optimization problem, 
$I\in\ \pref(P_{<j})$, where
$j\geq 1$, and let
$$
R_j[I] = 
 \{a>\top \LPifrank{j} \st  a\in I\} \cup 
 \{\neg a> \top \LPifrank{j}\st a\in \U\setminus I)\}. 
$$
Then 
\begin{enumerate}
\item $I\in\pref(P\cup R_j[I])$;
\item for every $J$ such that
$J\not=I$ and $I\geq^{P_{\leq j}} J$, $I>^{P\cup R_j[I]} J$.
\end{enumerate}
\end{lemma}
\begin{proof}
When proving (1), 
to simplify the notation, we write $R$ for $R_j[I]$. 
Since $I\in \pref(P_{<j})$, $I\in\mu(P)$. Clearly, 
$\mu(P\cup R)=\mu(P)$ and so, $I\in \mu(P\cup R)$. 
To show that $I\in\pref(P\cup R)$, let us consider an arbitrary
interpretation $J\in\mu(P\cup R)$ and assume that $J >^{P\cup R} I$. 
In particular, $J\not=I$ and so, $\dff^R(I,J)=j$.
If $\dff^P(I,J)<j$, then $\dff^{P\cup R}(I,J) <j$. Consequently, 
$J>^{(P\cup R)_{<j}} I$. Since all rules in $R$ are of rank $j$, 
it follows that $J>^{P_{<j}} I$, a contradiction with the fact that
$I\in\pref(P_{<j})$. Thus, $\dff^P(I,J)\geq j$. Since $\dff^R(I,J)=j$,
we have $\dff^{P\cup R}(I,J)=j$. Therefore, $J >^{P\cup R} I$ implies 
$J >^{R} I$, a contradiction again (since, by definition of
$R_j[I]=R$, $I\geq^R J$ for each interpretation $J$).
It follows that 
for every $J\in\mu(P\cup R)$, $J\not>^{P\cup R} I$, that
is, $I\in\pref(P\cup R)$. 

The assertion (2) is evident, since by definition of
$R_j[I]=R$, $I >^R J$ for each interpretation $J\neq I$.
\end{proof}

\begin{lemma}\label{lemma:selnew2}
Let $P$ be an optimization problem, $I,J$ interpretations such that
$I,J\in\pref(P_{<j})$, where 
$j\geq 1$, and let $R'_j[I,J] \in \L^{s,\geq j}$ be the union
of the  following sets of rules:
\[
\begin{array}{l}
 \{a>\top \LPifrank{j} \st  a\in I\cap J\}\\

 \{\neg a> \top \LPifrank{j}\st a\in \U\setminus (I\cup J))\}\\

 \{a\vee b > \top \LPifrank{j} \st a\in I\setminus J,\ b\in J\setminus I\} 
\\
 \{\neg a\vee \neg b > \top \LPifrank{j} \st a\in I\setminus J,\ b\in J\setminus I\} 
\\
\{(a\land b)\lor(\neg a\land \neg b) > \top  \LPifrank{j} \st 
a,b\in
       (I\setminus J)\cup( J\setminus I)\}.
\end{array}
\]
Then
\begin{enumerate}
\item for every $r\in R'_j[I,J]$, $v_I(r)=v_J(r)=1$;
\item if $I >^P J$, then $J\notin\pref(P\cup R'_j[I,J])$;
\item for every interpretation $K\notin \{I,J\}$, there is a
rule $r\in R'_j[I,J]$ such that $v_K(r)=2$;
\item if $I \not >^{P} J$, then $J\in \pref(P\cup R'_j[I,J])$.
\end{enumerate}
\end{lemma}
\begin{proof}
To simplify the notation, we write $R'$ for $R'_j[I,J]$. 
\smallskip

The assertion (1) is evident. Moreover, combined with $I >^P J$, it
yields $I>^{P\cup R'} J$. Thus, the assertion (2) holds. 
\smallskip

To prove the assertion (3), we consider 
an interpretation $K$ %
that is different from $I$ and from $J$, and we show
that there is a rule $r\in R'$ such that $v_K(r)=2$.

Let us consider $a\in I\cap J$. The rule $r_a= a>\top\LPifrank{j}$
belongs to $R'$. %
If $v_K(r_a)=2$, we are done. Thus, let
us assume that for every $a\in I\cap J$, $v_K(r_a)=1$, where
$r_a=a>\top\LPifrank{j}$. 
Consequently, for every $a\in I\cap J$, $a\in K$, that is,
$I\cap J\subseteq K$. Taking into account rules $\neg a>\top\LPifrank{j}$, 
with $a\in \U\setminus(I\cup J)$ and reasoning in the same way, we show 
that $K\subseteq I\cup J$.

Let us assume that $(I\cap J)\subset K \subset (I\cup J)$ holds. Then,
there is $a\in I\cup J$ such that $a\notin K$. Without loss of generality
we can assume that $a\in I$. From $(I\cap J)\subset K$, it follows that 
$a\notin J$. Further, there is $b\in K$ such that $b\notin I\cap J$.
Since $K \subset (I\cup J)$ holds, $b\in I\cup J$. It follows that
$a,b\in (I\setminus J)\cup (J\setminus I)$. Thus, $v_K(r)=2$, where $r$ 
is the corresponding rule from $R'$. 

Thus, only two possibilities for $K$ remain:
$K=I\cap J$ and $K=I\cup J$. It follows that 
$I\setminus J\neq \emptyset$ (otherwise, $I=K$) and $J\setminus I\neq
\emptyset$ (otherwise, $J=K$). Let $a'\in I\setminus J$ and $b'\in 
J\setminus I$. Then $R'$ contains rules $r= a'\lor b' >\top\LPifrank{j}$ 
and $s = \neg a'\lor \neg b' >\top\LPifrank{j}$. If $K=I\cap J$, $v_K(r)=2$. 
If $K=I\cup J$, $v_K(s)=2$.
\smallskip

To prove the assertion (4), we note first that (1) and $I \not >^{P} J$
together imply that $I\not >^{P\cup R'} J$. Next, we note that
if $\dff^P(J,K)<j$, then since $J\in\pref(P_{<j})$, $K\not>^{P\cup R'}
J$. If $\dff^P(J,K)\geq j$, then the property proved above implies that 
$K\not >^{P\cup R'}J$. Since $K$ is an arbitrary interpretation different 
from $I$ and $J$, and since $I\not>^{P\cup R'} J$, $J\in\pref(P\cup R')$
follows.
\end{proof}

We also note a property here that allows us to infer the strong 
sel-equivalence of two problems treated as CO problems from the 
strong sel-equivalence the these problems when treated as ASO 
problems (and conversely). The property relies on the fact that
changing selectors only does not affect the class of outcomes.
The proof is simple and we omit it.

\begin{lemma}
\label{lem:co-aso}
Let $P$ and $Q$ be optimization problems such that $\Mod{P^g}=\AS{P^g}$
and $\Mod{Q^g}=\AS{Q^g}$. Then, $P\equiv^{s,\geq i} Q$
($P\equiv^{s,=i} Q$, respectively), when $P$ and $Q$ are viewed as CO
problems, if and only if $P\equiv^{s,\geq i} Q$ ($P\equiv^{s,=i} Q$), when
$P$ and $Q$ are viewed as ASO problems.
\end{lemma}

The  final results in this section will be useful for
the complexity results.

\begin{lemma}\label{lem:compl-qbf-min-mod}
Deciding whether open QBFs
$\forall Y \phi(X,Y)$ 
$\forall Y \psi(X,Y)$ have the same minimal models 
is $\PiP{3}$-hard.
\end{lemma}
\begin{proof}
We show the result by a reduction from the $\PiP{3}$-hard problem of 
deciding satisfiability for QBFs of the form 
$\forall Z \exists X \forall Y \chi$.  
Let $\Phi$ be a QBF of such a form and
consider the following formulas, where $Z'=\{z' \mid z \in Z\}$, 
$u$ and $v$ are fresh atoms:
\begin{eqnarray*}
\phi& =& 
\bigwedge_{z\in Z} (z \leftrightarrow \neg z') 
\land 
\big( 
(\bigwedge_{x\in X} x \land u) 
\lor \chi\big)\land 
(v\lor \neg v)\\
\psi& =& \bigwedge_{z\in Z} (z \leftrightarrow \neg z') \land \big( (\bigwedge_{x\in X} x \land v) \lor \chi\big)\land (u\lor \neg u)
\end{eqnarray*}
Hence, the only difference between $\phi$ compared to $\psi$ is that 
we use $u$ and $v$ interchangably. Also note that the only point of including the conjuncts $v \lor \neg v$ and $u \lor \neg u$ is to have occurrences of $u$ and $v$ in both $\phi$ and $\psi$. We show that 
$\forall Y \phi(U,Y)$ and
$\forall Y \psi(U,Y)$ have the same minimal models 
(with open variables $U=Z\cup Z'\cup X \cup \{u,v\}$)
if and only if $\Phi$ is true.

Only-if direction: Assume $\Phi$ is false. Then, 
there exists an assignment $I\subseteq Z$, such 
that for all assignments to $X$, $\forall Y \chi$ is false. 
We show that 
$M_u=I\cup (Z\setminus I)'\cup X \cup \{u\}$ is a minimal model of 
$\forall Y\phi$.  
Indeed, it is a model of 
$\forall Y\phi$ and by the conjunction 
$\bigwedge_{z\in Z} (z \leftrightarrow \neg z')$ the only candidates for 
models $N\subset M_u$ of $\forall Y\phi$ are of the form
$I\cup (Z\setminus I)'\cup J$ with $J\subset X\cup \{u\}$, 
but then 
$\forall Y\chi$ 
(note that we can safely shift in $\phi$ the quantifier $\forall Y$ in front
of $\chi$ since $\chi$ hosts the only occurrences of atoms from $Y$ in $\phi$)
would be true under
$I \cup (J\setminus \{u\})$, a contradiction to our assumption. 
By essentially the same arguments, it can be shown
that
$M_v=I\cup (Z\setminus I)'\cup X \cup \{v\}$ is a minimal model of 
$\forall Y\psi$. Since $M_u\neq M_v$ 
we have shown that 
$\forall Y\phi$ 
and 
$\forall Y\psi$  
have different minimal models.

If-direction: Suppose
$\forall Y\phi$ 
and 
$\forall Y\psi$   have different minimal models.
The only relevant difference 
between $\phi$ and $\psi$ 
is the conjunction
$(\bigwedge_{x\in X} x \land u)$ in $\phi$, 
resp.\ 
$(\bigwedge_{x\in X} x \land v)$ in $\psi$. 
We can conclude that the different minimal models
are of the form 
$M_u=I\cup (Z\setminus I)'\cup X \cup \{u\}$ being a minimal model of  
$\forall Y\phi$
and
$M_v=I\cup (Z\setminus I)'\cup X \cup \{v\}$ being a minimal model of 
 $\forall Y\psi$.
Let us consider $M_u$. Since it is a minimal model, there
is no proper subset of $M_u$ which is a model of $\forall Y\phi$, 
in particular there is no $N\subset X \cup \{u\}$ such that 
$I\cup (Z\setminus I)' \cup N$ 
is a model of $\forall Y\phi$.
From this we observe that, for no such $N$, 
$I\cup (Z\setminus I)' \cup N$ is a model 
of 
$(\bigwedge_{x\in X} x \land u) 
\lor \chi\big)$, which implies 
that for no such $N$,
$I\cup (Z\setminus I)' \cup (N\setminus\{u\})$ is a model
of $\forall Y \chi$. But then $\Phi$ is false, since 
we have found an interpretation $I\subseteq Z$ such that 
for no $J\subseteq X$, $\forall Y \chi$ is true under $I\cup J$. 
\end{proof}

\begin{lemma}\label{lem:ptime-satdeg}
Given a ranked preference rule $r$, and an interpretation $I$, calculating $\degree{I}{r}$ can be done in polynomial time.
\end{lemma}
\begin{proof}
Initialize a variable $x$ with $1$. Starting from $i=1$ check whether
$I \models \headi{i}{r}$, and if so, set $x$ to $i$ and halt;
otherwise increment $i$ and continue checking; if no more heads exist
in $r$, halt. Each of the checks is a model checking task for a
propositional formula and hence in polynomial time. Upon halting, $x$
is equal to $\degree{I}{r}$.
\end{proof}

\begin{lemma}\label{lem:ptime-diff}
Given an optimization problem $P$ and two interpretations $I,J$,
calculating $\dff^P(I,J)$ can be done in polynomial time.
\end{lemma}
\begin{proof}
Initialize a variable $x$ with $\infty$, scan the rules in $P^s$ and
for each ranked preference rule $r\in P^s$, determine whether $v_I(r)
\neq v_J(r)$ (in polynomial time due to Lemma~\ref{lem:ptime-satdeg}).
If so, set $x$ to $\rank{r}$ if $\rank{r} < x$. After having processed
all rules, $x$ is equal to $\dff^P(I,J)$.
\end{proof}

\begin{lemma}\label{lem:ptime-pref-check}
Given an optimization problem $P$, and two interpretations $I,J$, 
deciding whether $I>^P J$ can be done in polynomial time.
\end{lemma}
\begin{proof}
First, sort the rules in $P^s$ by their ranks.  Starting from the
lowest rank upwards, do the following for each rank $i$: Check for all
rules of rank $i$ whether $v_I(r) < v_J(r)$ or $v_I(r) \leq
v_J(r)$. If $v_I(r) < v_J(r)$ holds at least for one rule and $v_I(r)
\leq v_J(r)$ for all other rules of rank $i$, accept.  If there are
rules $r$ and $r'$ of the rank $i$ such that $v_I(r) < v_J(r)$ and
$v_I(r') > v_J(r')$, reject. If all ranks have been processed,
reject. By Lemma~\ref{lem:ptime-satdeg}, all steps are doable in
polynomial time.
\end{proof}

\begin{lemma}\label{lem:conp-prefcheck-member}
Given a classical optimization problem $P$ and an interpretation $I$,
deciding whether $I\in\pref(P)$ is in $\CONP$.
\end{lemma}
\begin{proof}
We show that a witness $J$ for the complementary problem (deciding
whether $I\notin\pref(P)$) can be verified in polynomial time.  If
$J=I$, verify in polynomial time that $I$ does not satisfy the
propositional theory $P^g$, which is well-known to be feasible in
polynomial time. Otherwise, verify in polynomial time that $J$
satisfies $P^g$ and that $J >^P I$ (both in polynomial time, the
latter by Lemma~\ref{lem:ptime-pref-check}).
\end{proof}

\begin{lemma}\label{lem:pip2-prefcheck-member}
Given an answer set optimization problem $P$ and an interpretation $I$,
deciding whether $I\in\pref(P)$ is in $\SigmaP{2}$.
\end{lemma}
\begin{proof}
We show that a witness $J$ for the complementary problem (deciding
whether $I\notin\pref(P)$) can be verified in polynomial time using an
$\NP$ oracle. If $J=I$, verify that $I$ does not satisfy $P^g$ using
the $\NP$ oracle. This is possible because answer-set checking is
$\CONP$-complete (Theorem~8 of \cite{PearceTW09}). Otherwise, verify
using the $\NP$ oracle that $J$ satisfies the propositional theory
$P^g$ and that $J >^P I$ (in polynomial time by
Lemma~\ref{lem:ptime-pref-check}).
\end{proof}

\section*{Appendix B: Proofs}

\thrm{thm:sequivrankedprefeq}{}
{For every ranked optimization problems $P$ and $Q$, and every rank intervals
$[i,j]$,
$P \equiv^{s,[i,j]} Q$ if and only if the following
conditions hold:
\begin{enumerate}
\item $\pref(P_{<i})=\pref(Q_{<i})$
\item ${>^P_{\pref(P_{<i})}}={>^Q_{\pref(Q_{<i})}}$
\item For every $I,J\in \pref(P_{<i})$ such that $i<\dff^P(I,J)$ or 
$i<\dff^Q(I,J)$, $\dff^P(I,J)=\dff^Q(I,J)$ or both $\dff^P(I,J)>j$ and
$\dff^Q(I,J)>j$.
\end{enumerate}
}

\smallskip
\noindent
\begin{proof}
($\Leftarrow$)
Let $R\in\L^{s,[i,j]}$ and let $I\in \pref(P\cup R)$. By 
Lemma~\ref{lemma:lowerrank}, $I\in \pref((P\cup R)_{<i})$. Since 
$R \in\L^{s,[i,j]}$, $R_{<i}=(\emptyset,\emptyset)$. Thus, $I\in
\pref(P_{<i})$. By assumption, it follows that $I\in\pref(Q_{<i})$.
In particular, $I\in\acc(Q_{<i})$ and, as $\acc(Q_{<i})=\acc(Q)$, 
$I\in\acc(Q)$. Since $R^g=\emptyset$, we have $I\in\acc(Q\cup R)$. 
To show that $I\in\pref(Q\cup R)$ we have to show that there is no 
$J\in \acc(Q\cup R)$ such that $J >^{Q\cup R} I$. Let us assume to 
the contrary that such a $J$ exists. By Lemma ~\ref{lemma:compose}, 
there are three possibilities.

First, we assume that $\dff^Q(I,J)<\dff^R(I,J)$ and $J>^Q I$. The 
latter property implies $\dff^Q(I,J)\geq i$ (otherwise, we would have 
$J>^{Q_{<i}} I$, contrary to $I\in\pref(Q_{<i})$). In particular,
we have $I\approx^{Q_{<i}} J$ and, since $I\in\pref(Q_{<i})$, it 
follows that $J\in\pref(Q_{<i})$. By (1), $J\in \pref(P_{<i})$. Thus, by
(2), $J>^P I$.
If $\dff^Q(I,J) \geq j$ then $\dff^R(I,J)>j$ and, as $R \in\L^{s,[i,j]}$,
$\dff^R(I,J)=\infty$. Since $J>^P I$, $J>^{P\cup R}I$. Otherwise, 
$i\leq \dff^Q(I,J) <j$. If $i < \dff^Q(I,J)$ then, by (3), $\dff^P(I,J)
=\dff^Q(I,J)$. If $i=\dff^Q(I,J)$, then again by (3), $\dff^P(I,J)\leq i$.
In either case, $\dff^P(I,J)<\dff^R(I,J)$. Since $J>^P I$, $J>^{P\cup R}I$.

Next, let us assume that $\dff^Q(I,J)>\dff^R(I,J)$ and $J>^R I$. Since 
$R\in\L^{s,[i,j]}$, it follows that $\dff^R(I,J)\geq i$ and so,
$\dff^Q(I,J)>i$. We recall that $I\in \pref(Q_{<i})$. Thus, $J\in
\pref(Q_{<i})$ and, consequently, $J\in \pref(P_{<i})$. If $\dff^Q(I,J)
\leq j$ then, by (3), $\dff^P(I,J)=\dff^Q(I,J)$ and so, 
$\dff^P(I,J)>\dff^R(I,J)$. If $j<\dff^Q(I,J)$ then, also by (3), 
$j<\dff^P(I,J)$. 
Since $\dff^Q(I,J)>\dff^R(I,J)$, $\dff^R(I,J)<\infty$ and, consequently, 
$\dff^R(I,J)\leq j$. Thus, $\dff^P(I,J)>\dff^R(I,J)$ in this case, too.
Since $J>^R I$, $J>^{P\cup R} I$ follows. 

Finally, let us assume that $\dff^Q(I,J)=\dff^R(I,J)$, $J>^Q I$
and $J>^R I$. Since $R\in\L^{s,[i,j]}$, $\dff^R(I,J)\geq i$. Thus,
$\dff^Q(I,J)\geq i$ and, since $I\in\pref(Q_{<i})$, $J\in\pref(Q_{<i})$.
By (1) we have $J\in\pref(P_{<i})$ and, by (2), $J>^P I$. Consequently, 
$J>^{P\cup R} I$.

In all cases we obtained $J>^{P\cup R} I$, contrary to $I\in\pref(P\cup R)$,
a contradiction.

\smallskip
\noindent
($\Rightarrow$)
Let us assume that $\pref(P_{<i})\neq \pref(Q_{<i})$.
Without loss of generality, we can assume that there is
$I\in \pref(P_{<i}) \setminus \pref(Q_{<i})$ and define
$R=(\emptyset,R_i[I])\in\L^{s,=i}$, where $R_i[I]$ is as in
Lemma~\ref{lemma:selnew}. By that lemma, $I\in \pref(P\cup R)$.
On the other hand, since $I\notin \pref(Q_{<i})$ and $R\in\L^{s,=i}$,
$I\notin \pref((Q\cup R)_{<i})$. By Lemma~\ref{lemma:lowerrank}, $I\notin 
\pref(Q\cup R)$. Thus, $P \not\equiv^{s,=i} Q$, contrary to the
assumption.

It follows that $\pref(P_{<i})= \pref(Q_{<i})$, that is, that the
condition (1) holds. To prove the condition (2), let us consider
interpretations $I,J\in\pref(P_{<i})$ such that $I>^P J$. Let
$R'_i[I,J]$ be the selector defined in Lemma~\ref{lemma:selnew2}.
Since $I>^P J$, Lemma~\ref{lemma:selnew2}(2) implies that $J\notin
\pref(P\cup R'_i[I,J])$. Consequently, $J\notin \pref(Q\cup R'_i[I,J])$.
By Lemma~\ref{lemma:selnew2}(4), it follows that $I>^{Q} J$.
By symmetry, $I>^Q J$ implies $I>^P J$ and so, the condition (2) holds.

To prove the condition (3), let us assume that there are interpretations
$I$ and $J$ that satisfy the assumptions but violate the corresponding 
conclusion. In what follows, we write $p$ for $\dff^P(I,J)$ and $q$ for 
$\dff^Q(I,J)$. Thus, we have $p>i$ or $q>i$, $p\not=q$, and $p\leq j$
or $q\leq j$. Without loss of generality, we can assume that $p < q$. 
It follows that $p$ is finite and, consequently, that $\dff^P(I,J)<
\infty$ and $I\not= J$. Moreover, $i<q$ and $p\leq j$.

Let us assume first that $p < i$. We take the problem 
$R=(\emptyset,R_i[J])$, where $R_i[J]$ is as specified in
Lemma~\ref{lemma:selnew} and define $P'=P\cup R$ and $Q'=Q\cup R$.
Since $I,J\in\pref(P_{<i})$ we also have $I,J\in\pref(Q_{<i})$.
By our assumptions, $q>i$. Thus, $J\approx^{Q_{\leq i}} I$ and,
in particular, $J\geq^{Q_{\leq i}} I$.  We recall that $I\not=J$.
Consequently, by the assertion (2) of Lemma \ref{lemma:selnew}
we have that $J >^{Q'} I$.
Since all rules in $R$ have ranks $i$, we have $I,J\in\pref(P'_{<i})$ 
and $\dff^{P'}(I,J)=\dff^P(I,J)<i$. It follows that $J\not>^{P'} I$ 
(otherwise, by $\dff^{P'}(I,J)<i$ we would have $J>^{P'_{<i}} I$).
Let us define $R'=(\emptyset,R'_i[I,J])$, where $R'_i[I,J]$ is as
specified in Lemma~\ref{lemma:selnew2}. Since $J\not>^{P'} I$, by the
assertion (4) of that lemma, $I\in\pref(P'\cup R')$.
We have $P'\cup R'=P\cup(R\cup R')$. Thus, $I\in\pref(P\cup (R\cup R'))$
and, by $P\equiv^{s,[i,j]} Q$, $I\in\pref(Q\cup (R\cup R'))=\pref(Q'\cup R')$.
By the assertion (2) of Lemma~\ref{lemma:selnew2}, $J\not>^{Q'} I$,
a contradiction.

Next, let $p=i$. Clearly, $I\not>^P J$ or $J\not>^P I$. Without loss 
of generality, let us assume that $J\not>^P I$. Let $R'=(\emptyset,
R_i[I,J])$, and let us define $P'=P\cup R'$ and $Q'=Q\cup R'$. Since
all rules in $R'$ have ranks $i$, $I,J\in\pref(P_{<i})$ implies
$I,J\in\pref(P'_{<i})$. Moreover, from $J\not>^P I$ it follows by
Lemma \ref{lemma:selnew2}(1) that $J\not>^{P'} I$. 
Let $R=(\emptyset,R_i[J])$. All rules in $R$ have rank $i$, 
$\dff^{P'}(I,J)=\dff^P(I,J)=i$ (the first equality holds by Lemma
\ref{lemma:selnew2}(1) and $J\not>^{P'} I$. Thus, it follows that 
$J\not>^{P'\cup R} I$. Moreover, for every $K\notin\{I,J\}$, if 
$\dff^{P'}(K,I)<i$, then $K\not>^{P'\cup R} I$ follows from $I\in 
\pref(P'_{<i})$. If $\dff^{P'}(K,I)\geq i$, then $K\not>^{P'\cup R} I$
follows from Lemma \ref{lemma:selnew2}(3). Thus, $I\in \pref(P'\cup R)$.
On the other hand, we recall that $\dff^Q(I,J)=q>i$.  Thus, 
$\dff^{Q'}(I,J)>i$, too (Lemma \ref{lemma:selnew2}(1)). It follows that
$J\geq^{Q'_{\leq i}} I$. Consequently, by Lemma \ref{lemma:selnew}(2), 
we have $J>^{Q'\cup R} I$. Thus, $I\notin\pref(Q'\cup R)$,
a contradiction.

It follows that $p>i$. To complete the proof of (3), we recall that 
$p\leq j$. Clearly, $I\not>^P J$ or $J\not>^P I$. Without loss of 
generality, let us assume that $J\not>^P I$. Let 
$R'=(\emptyset,R_i[I,J])$, and let us define $P'=P\cup R'$ and 
$Q'=Q\cup R'$. Let us assume that for some interpretation $K\notin\{I,J\}$,
$K >^{P'} I$. By Lemma \ref{lemma:selnew2}(3), it follows that 
$\dff^{P'}(I,K)< i$. Thus, $\dff^{P}(I,K)< i$, a contradiction with
$I\in \pi(P_{<i})$. Thus, for every interpretation $K\notin\{I,J\}$, 
$K\not>^{P'} I$ and, by the same argument, $K\not>^{P'} J$. 
Consequently, for every interpretation $K\notin\{I,J\}$, $K\not>^{P'_{<p}}
I$ and $K\not>^{P'_{<p}} J$. In addition, since $I,J\in\pi(P_{<i})$,
by Lemma \ref{lemma:selnew2}(1) we obtain that neither $I>^{P'_{<p}} J$
nor $J>^{P'_{<p}} I$. Thus, $I,J\in\pref(P'_{<p})$. In addition, by Lemma 
\ref{lemma:selnew2}(1), $\dff^{P'}(I,J)=\dff^P(I,J)=p$ and, since 
$J\not>^P I$, $J\not>^{P'} I$.

Let $R=(\emptyset,R_p[J])$. As $\dff^{P'}(I,J)=p$, $J\not>^{P'} I$,
and all rules in $R$ have rank $p$, it follows that $J\not>^{P'\cup 
R} I$. Moreover, for every $K\notin\{I,J\}$, if $\dff^{P'}(K,I)<i$, 
then $K\not>^{P'\cup R} I$ follows from $I\in \pref(P'_{<i})$. If 
$\dff^{P'}(K,I)\geq i$, then $K\not>^{P'\cup R} I$ follows from Lemma 
\ref{lemma:selnew2}(3) (and the definition of $P'$). Thus, $I\in
\pref(P'\cup R)$. On the other hand, we recall that $\dff^Q(I,J)=q>p$.  
Thus, $\dff^{Q'}(I,J)>p$, too (Lemma \ref{lemma:selnew2}(1)). It follows
that $J\geq^{Q'_{\leq p}} I$. Consequently, by Lemma \ref{lemma:selnew}(2),
we have $J>^{Q'\cup R} I$. Thus, $I\notin\pref(Q'\cup R)$,
a contradiction (we recall that $P\equiv^{s,[i,j]}Q$, $P'=P\cup(R\cup R')$, 
$Q'=Q\cup (R\cup R')$, and, as $p\leq j$, $R\cup R' \in \L^{s,[i,j]}$).
\end{proof}

\thrm{thm:simpleg}{}
{For every two CO (ASO, respectively) problems $P$ and $Q$, $P \equiv_g Q$
if and only if $P^g$ and $Q^g$ are strongly equivalent (that is,
$\Mod{P^g}=\Mod{Q^g}$ for CO problems, and $\HT{P^g} = \HT{Q^g}$
for ASO problems) and ${>^P_{{\Mod{P^g}}}}= {>^Q_{{\Mod{Q^g}}}}$.
}

\smallskip
\noindent
\begin{proof}
($\Leftarrow$)
The first assumption implies strong equivalence of the generators 
$P^g$ and $Q^g$ relative to the corresponding semantics (we recall 
that in the case of classical semantics, strong and standard equivalence 
coincide). It follows that for
every problem $R\in\L^g$, $\mu(P\cup R)=\mu(Q\cup R)$. Moreover, for each
semantics, $\mu(P\cup R)\subseteq\Md(P^g\cup R^g)\subseteq\Md(P^g)$ and,
similarly, $\mu(Q\cup R)\subseteq\Md(Q^g\cup R^g)\subseteq\Md(Q^g)$.
Since ${>^P_{\Mod{P^g}}}= {>^Q_{\Mod{Q^g}}}$, and $R\in\L^g$ does not
change preferences, we have ${>^{P\cup R}_{\mu(P^g\cup R^g)}} =
{>^{Q\cup R}_{\mu(Q^g\cup R^g)}}$. By Lemma \ref{lemma:simple},
$\pref(P\cup R)=\pref(Q\cup R)$.

\smallskip
\noindent
($\Rightarrow$) Let us assume that $P^g$ and $Q^g$ are not strongly
equivalent. Then, there is a problem $R\in\L^g$ such that $\mu(P\cup R)
\neq \mu(Q\cup R)$. Without loss of generality, we can assume that for
some interpretation $I$, $I\in\mu(P\cup R) \setminus \mu(Q\cup R)$. Let
us define a problem $T\in \L^g$ by setting $T=(\Pi[I],\emptyset)$, where
$\Pi[I]$ is as defined in Lemma \ref{lemma:as-elim1}. By that lemma, 
$\mu(P\cup R\cup T) =\{I\}$ and $\mu(Q\cup R\cup T) =\emptyset$. The 
former property implies that $I$ is necessarily preferred, that is
$I \in \pref(P \cup R \cup T)$, and the latter one implies that
$\pref(Q \cup R \cup T)=\emptyset$. This is a contradiction with the
assumption that $P \equiv_g Q$. Thus, $P^g$ and $Q^g$ are strongly
equivalent, that is, $\Md(P^g)=\Md(Q^g)$, in the case $P$ and $Q$ are
CO problems, and $\HT{P^g}=\HT{Q^g}$, in the case $P$ and $Q$ are
ASO problems.

Since $\HT{P^g} = \HT{Q^g}$ implies $\Md(P^g)=\Md(Q^g)$, the identity
holds in each of the two cases. Because of the equality, we will write 
$M$ for both $\Md(P^g)$ and $\Md(Q^g)$. It remains to show that $>^P_{M} =
>^Q_{M}$. Towards a contradiction, let us assume that there are $I,J
\in M$ that are in exactly one of these two relations; without loss of
generality we will assume that $I >^P J$ and $I \not>^Q J$. The former
identity implies, in particular, that $I\not=J$. Let $T=(\Pi[I,J],
\emptyset)$, where $\Pi[I,J]$ is a theory defined in
Lemma \ref{lemma:as-elim2}. By that lemma, $\mu(P\cup T) =
\mu(Q\cup T)=\{I,J\}$. Clearly, $J\notin\pref(P\cup T)$
and $J\in\pref(Q\cup T)$, contrary to our assumption that $P\equiv_g Q$.
\end{proof}

\nop{
\smallskip
\noindent
\begin{proof}
For the ``if'' direction,
let us assume 
$
\Mod{P^g}=\Mod{Q^g}\ \mbox{and}\ {>^P_{\Mod{P^g}}}=
{>^Q_{\Mod{Q^g}}},
$
and let us consider 
a generator $R\in\L^g$. We note that
$\acc(P\cup R)=\Mod{P^g\cup R^g}=\Mod{Q^g\cup R^g}=\acc(Q\cup R)$.
Indeed, the first and the third equalities follow from the fact that
we are considering classical optimization problems, and the second 
equality follows from the assumption $\Mod{P^g}=\Mod{Q^g}$ and the
property of classical logic that for every two theories $T$ and $T'$,
$\Mod{T\cup T'}=\Mod{T}\cap\Mod{T'}$.
Moreover, since ${>^P_{\Mod{P^g}}}= {>^Q_{\Mod{Q^g}}}$ and
$R\in\L^g$ does not change preferences, we have
${>^{P\cup R}_{\Mod{P^g}}} = 
{>^{Q\cup R}_{\Mod{Q^g}}}$. 

We now use the properties
$\Mod{Q^g\cup R^g}=\Mod{P^g\cup R^g}\subseteq  \Mod{P^g}=\Mod{Q^g}$,
$\acc(P\cup R)=\Mod{P^g\cup R^g}$, and $\acc(Q\cup R)=\Mod{Q^g\cup R^g}$,
to arrive at
${>^{P\cup R}_{\acc(P\cup R)}} =
{>^{Q\cup R}_{\acc(Q\cup R)}}$. Applying Lemma \ref{lemma:simple} 
we get $\pref(P\cup R) =\pref(Q\cup R)$, as needed.

For the ``only-if'' direction, let us first suppose that 
$\Mod{P^g}\neq \Mod{Q^g}$. Without loss of generality, 
we can assume
that there is  $I\in\Mod{P^g}\setminus \Mod{Q^g}$. We set $R=(\Pi[I],\emptyset)$,
where $\Pi[I]$ is as defined in Lemma~\ref{lemma:as-elim1}. %
By that lemma,
$\acc(P\cup R)=\Mod{P^g\cup R^g}=\{I\}$ while
$\acc(Q\cup R)=\Mod{Q^g\cup R^g}=\emptyset$. 
It follows that
$\pref(P\cup R)=\{I\}$ and 
$\pref(Q\cup R)=\emptyset$. Consequently, we have
$P \not \equiv_g Q$, a contradiction.

Thus, $\Mod{P^g}= \Mod{Q^g}$ holds. From now on in the proof, we write $M$ for 
 $\Mod{P^g}$ ($= \Mod{Q^g}$). Let us assume that
$>^P_M\;\neq\; >^Q_M$. It follows that there exist 
$I,J\in M$ that are in exactly one of the relations
$>^P_M$ and $>^Q_M$. Without loss of generality, we can assume
that $I>^P_M J$ and $I\not>^Q_M J$. 
We define $R=(\Pi[I,J],\emptyset)$, where $\Pi[I,J]$ is as described
in Lemma~\ref{lemma:as-elim2}.  %
We obtain 
$\acc(P\cup R)=\Mod{P^g\cup R^g}=\{I,J\}
=\Mod{Q^g\cup R^g}=\acc(Q\cup R)$. 
Moreover, since $R\in\L^g$, the identities 
$I >^{P\cup R}_{\acc(P\cup R)} J$ 
and  
$I \not >^{Q\cup R}_{\acc(Q\cup R)} J$ are maintained.
By definition, $J\in\pref(Q\cup R)$ but $J\notin \pref(P\cup R)$.
Thus, $\pref(P\cup R)\neq \pref(Q\cup R)$ and, consequently, $P \not
\equiv_g Q$. This contradiction completes the proof.
\end{proof}

\thrm{thm:char-aso-sequivgen}{}{
For every ranked ASO problems $P$ and $Q$, $P \equiv_g Q$ if and only
if $\HT{P^g} = \HT{Q^g}$ and ${>^P_{\Md(P^g)}} = {>^Q_{\Md(Q^g)}}$.
}

\smallskip
\noindent
\begin{proof}
($\Leftarrow$)
  Let $R \in \L^g$. We have 
  $\mu(P \cup R) = 
   \AS{P^g \cup R^g} = 
   \AS{Q^g \cup R^g} = 
   \mu(Q \cup R)$, the second equality ensured by the assumption
   $\HT(P^g) = \HT(Q^g)$, which implies strong equivalence of $P^g$ and
$Q^g$.
 Since $\HT(P^g) = \HT(Q^g)$, also $\Md(P^g) = \Md(Q^g)$. It follows that 
$\Md(P^g \cup R^g) = \Md(Q^g \cup R^g)$. Moreover, since $\Md(P^g\cup R^g)
\subseteq \Md(P^g)$ and $\Md(Q^g\cup R^g) \subseteq
 \Md(Q^g)$, we have $>^P_{\Md(P^g \cup R^g)} = >^Q_{\Md(Q^g \cup R^g)}$.  Finally,
	since $R^s=\emptyset$, 
$>^{P\cup R}_{\Md(P^g \cup R^g)} = >^{Q\cup R}_{\Md(Q^g \cup R^g)}$.
Thus, it follows that $\pref(P \cup R) = \pref(Q \cup R)$.

\smallskip
\noindent
($\Rightarrow$) Let us assume that $\HT(P^g) \neq \HT(Q^g)$. Then, there
is a problem $R\in\L^g$ such that $\AS{P^g \cup R^g} \neq \AS{Q^g \cup R^g}$.
Without loss of generality, we can assume that for some interpretation 
$I$, $I \in \AS{P^g \cup R^g} \setminus \AS{Q^g \cup R^g}$. Let us 
define a problem $T\in \L^g$ by setting $T=(\Pi[I],\emptyset)$, where 
$\Pi[I]$ is as defined in Lemma \ref{lemma:as-elim1}.
By that lemma, $\AS{P^g \cup R^g\cup T^g} =\{I\}$ and $\AS{Q^g \cup R^g 
\cup T^g} =\emptyset$. 
The former property implies that $I$ is necessarily preferred, that is
  $I \in \pref(P \cup R \cup R')$, 
and the latter one implies that $\pref(Q^g \cup R \cup R')=\emptyset$.
  This is a contradiction 
with the assumption that $P \equiv_g Q$. 
  Thus, $\HT(P^g) = \HT(Q^g)$ follows.

\smallskip
\noindent
  Since $\HT(P^g) = \HT(Q^g)$, we have $\Md(P^g)=\Md(Q^g)$. We will write
  $M$ for that set. It remains to show that $>^P_{M} = >^Q_{M}$. Towards
  a contradiction, let us assume that there are $I,J\in M$ that are in
  exactly one of 
  these two relations; without loss of generality we will assume that
  $I >^P J$ and $I \not>^Q J$. 
  Let $T=(\Pi[I,J],\emptyset)$, where $\Pi[I,J]$ is a theory defined in
  Lemma \ref{lemma:as-elim2}. By that lemma, $\AS{P^g \cup T^g} = 
  \AS{Q^g \cup T^g}=\{I',J'\}$. Since $t$ and $t'$ are fresh, we also
  have $I' >^P J'$, and $I' \neg>^Q J'$. Clearly, $J'\notin\pref(P\cup T)$
  and $J'\in\pref(Q\cup T)$, contrary to our assumption that $P\equiv_g Q$.
\end{proof}
}{}

\thrm{thm:cequivrankedprefeqcomb}{}
{
For every ranked CO (ASO, respectively) problems $P$ and $Q$, 
and every rank intervals $[i,j]$,
$P \equiv^{s,[i,j]}_g Q$ if and only if the following conditions hold:
\begin{enumerate}
\item $P^g$ and $Q^g$ are strongly equivalent (that is, $\Mod{P^g}=
\Mod{Q^g}$ for CO problems, and $\HT{P^g}=\HT{Q^g}$ for ASO problems)
\item ${>^P_{\Mod{P^g}}}={>^Q_{\Mod{Q^g}}}$
\item For every $I,J\in \Md(P^g)$ such that $i<\dff^P(I,J)$ or
$i<\dff^Q(I,J)$, $\dff^P(I,J)=\dff^Q(I,J)$ or both $\dff^P(I,J)>j$ and
$\dff^Q(I,J)>j$
\item ${>^{P_{<i}}_{\Mod{P^g}}}= {>^{Q_{<i}}_{\Mod{Q^g}}}$.
\end{enumerate}
}
\noindent
\begin{proof}
The condition (1) is equivalent to the property that $P^g$ and $Q^g$
are strongly equivalent relative to the corresponding semantics. Using
this observation, we will provide a single argument for the two versions
of the assertion.

\smallskip 
\noindent
($\Leftarrow$) By Proposition \ref{prop:ortho}, it suffices to prove that 
for every $R\in\L^g$, $P\cup R \equiv^{s,[i,j]} Q\cup R$. By our comment 
above, we have that $\mu(P\cup R)=\mu(Q\cup R)$ (we recall that $\mu(P)$ 
denotes 
the set of outcomes of an optimization problem $P$; $\mu(P)=\Md(P^g)$ in 
the case of CO problems, and $\mu(P)=\AS{P^g}$ in the case of ASO 
problems). Moreover, for each type of problems, we also have 
$\mu(P\cup R)\subseteq
\Mod{P^g\cup R^g}\subseteq \Mod{P^g}$ and, similarly, $\mu(Q\cup R)\subseteq
\Mod{Q^g\cup R^g}\subseteq\Mod{Q^g}$. Since all rules in $R$ have rank
at least $i$, by the condition (4) it follows that $>^{(P\cup R)_{<i}}_{\mu(P\cup R)} 
= >^{(Q\cup R)_{<i}}_{\mu(Q\cup R)}$.
By Lemma~\ref{lemma:simple}, %
we have $\pref((P\cup R)_{<i})=\pref((Q\cup R)_{<i})$ and so, the 
condition (1) of Theorem \ref{thm:sequivrankedprefeq} holds for $P\cup R$ 
and $Q\cup R$. Since $\pref((P\cup R)_{<i}) \subseteq \mu(P\cup R)
\subseteq\Mod{P^g}$, and since the corresponding inclusions hold for $Q$,
too, the conditions (2)--(3) of this theorem for $P$ and $Q$ imply the 
conditions (2)--(3) from Theorem \ref{thm:sequivrankedprefeq} for $P\cup R$ 
and $Q\cup R$. Thus, by Theorem \ref{thm:sequivrankedprefeq}, $P\cup R 
\equiv^{s,[i,j]} Q\cup R$.

\smallskip
\noindent
($\Rightarrow$) Let us assume that $P\equiv_g^{s,[i,j]} Q$. Then,
$P\equiv_g Q$ follows and, by Theorem \ref{thm:simpleg}, implies the 
appropriate version of the condition (1). Since $\HT{P^g}=\HT{Q^g}$ 
implies $\Mod{P^g}=\Mod{Q^g}$, for each of the two versions of the 
assertion we have $\Mod{P^g}=\Mod{Q^g}$. From now on in the proof, 
we write $M$ for $\Mod{P^g}$ and, because of the equality, also for 
$\Mod{Q^g}$.

Next, for interpretations $I,J\in M$, $I\not=J$, we define $R=(\Pi[I,J],
\emptyset)$, where $\Pi[I,J]$ is as in Lemma \ref{lemma:as-elim2}. Let 
us define $P_1=P\cup R$ and $Q_1=Q\cup R$. We have $P_1\equiv^{s,[i,j]}
Q_1$. Moreover, by Lemma \ref{lemma:as-elim2}, we also have that $\mu(P_1)=
\mu(Q_1)=\{I,J\}$.

To prove the condition (4), let us assume that $I>^{P_{<i}} J$. It 
follows that $I>^{P_1} J$ (we recall that $R$ contains no preference 
rules). Since $\mu(P_1)=\{I,J\}$, $J\notin\pref(P_1)$ and 
$I\in\pi(P_1)$.
By the assumption, $J\notin\pref(Q_1)$. Since $\mu(Q_1)=\{I,J\}$, we
have that $I>^{Q_1} J$. In particular, $I\in\pi(Q_1)$. 
If $\dff^Q(I,J)<i$ then, since $R$ has no 
preference rules, $I>^{Q_{<i}} J$.
Thus, let us assume that $\dff^Q(I,J)\geq i$ and let us define $R'=
(\emptyset, R_i[J])$, where $R_i[J]$ is as in Lemma \ref{lemma:selnew}. 
Since $I>^{P_{<i}} J$, $R$ has no preference rules and all preference
rules in $R'$ have rank $i$, it follows that $I>^{P_1\cup R'} J$. The
generator module in $R'$ is empty. It follows that $\mu(P_1\cup R')= 
\mu(Q_1\cup R')=\{I,J\}$. Thus, $J\notin\pref(P_1\cup R')$ and, 
consequently, $J\notin\pref(Q_1\cup R')$. It follows that $I>^{Q_1\cup
 R'} J$. Since $\dff^Q(I,J)\geq i$, $\dff^{Q_1}(I,J)\geq i$. 
Thus, $J\in\pref{(Q_1)_{<i}}$. By Lemma \ref{lemma:selnew}, $J\in\pref(Q_1
\cup R')$, a contradiction. The argument shows that $I>^{P_{<i}} J$ 
implies $I>^{Q_{<i}} J$. The converse implication follows by symmetry 
and so, the condition (4) holds.

To prove the condition (2), let us assume that $I>^{P} J$. If
$\dff^P(I,J)<i$, then $I>^{P_{<i}} J$ and, by (4), $I>^{Q_{<i}} J$.
Thus, $I>^{Q} J$. Let us assume then that $\dff^P(I,J) \geq i$. Since
$\mu(P_1)=\{I,J\}$ and since $I>^{P} J$ implies $I>^{P_1} J$, $J
\notin\pref(P_1)$. Thus, $J\notin\pref(Q_1)$. Since $\mu(Q_1)=\{I,J\}$,    
$I>^{Q_1} J$ and so, $I>^{Q} J$.

To prove the condition (3), without loss of generality we assume that
$i <\dff^P(I,J)$. Thus, we also have $i<\dff^{P_1}(I,J)$ and that 
$I,J\in\pref((P_1)_{<i})$, the latter follows from the properties that 
$\mu(P_1)=\{I,J\}$ and that $\dff^{P_1}(I,J)=\dff^P(I,J)>i$. 
Since $P_1\equiv^{s,[i,j]} Q_1$, the condition (3) of Theorem 
\ref{thm:sequivrankedprefeq} holds for $P_1$, $Q_1$, $I$ and $J$, that 
is, $\dff^{P_1}(I,J)=\dff^{Q_1}(I,J)$ or both $\dff^{P_1}(I,J)>j$ and
$\dff^{Q_1}(I,J)>j$. Consequently, $\dff^{P}(I,J)=\dff^{Q}(I,J)$ or both
$\dff^{P}(I,J)>j$ and $\dff^{Q}(I,J)>j$, that is, the condition (3) holds. 
\end{proof}

\thrm{thm:co-sel:simple}{}{
Given optimization problems $P$ and $Q$, 
deciding $P\equiv^{s} Q$ is $\CONP$-complete in case of CO-problems 
and $\PiP{2}$-complete in case of ASO-problems.
}

\smallskip
\noindent
\begin{proof}
By Corollary~\ref{cor:sequivrankedpref}, $P\equiv^{s} Q$
iff 
jointly
(1) $\mu(P)=\mu(Q)$,
(2) for every $I,J\in \mu(P)$, $\dff^P(I,J) =\dff^Q(I,J)$, and
(3) ${>^P_{\mu(P)}}={>^Q_{\mu(Q)}}$.

For membership, we consider the complementary problem, that is, deciding whether
(a) $\mu(P)\neq\mu(Q)$, or
(b) for some $I,J\in \mu(P)$, $\dff^P(I,J) \neq \dff^Q(I,J)$, or
(c) ${>^P_{\mu(P)}}\neq{>^Q_{\mu(Q)}}$,
and show that it is in $\NP$ for CO-problems and in $\SigmaP{2}$ for
ASO-problems. To do this, we show that a witness consisting of two
interpretations can be verified in polynomial time, which in case of
ASO-problems uses an $\NP$ oracle.

Consider therefore such a witness $(I,J)$ of two interpretations. We
first test for $I\in\mu(P)$, $I\in\mu(Q)$, $J\in\mu(P)$, and
$J\in\mu(Q)$. In case of CO problems, each of these tests corresponds
to model checking in classical propositional logic, which is
well-known to be feasible in polynomial time. In case of ASP problems,
we require one call each to an $\NP$ oracle, since answer-set checking
is $\CONP$-complete (Theorem~8 of \cite{PearceTW09}).  If $I\in \mu(P)
\Delta \mu(Q)$ or $J\in \mu(P) \Delta \mu(Q)$ (where $\Delta$ denotes
the symmetric difference), condition (a) is satisfied and the
algorithm accepts. If $I \notin \mu(P) \cap \mu(Q)$ or $I \notin
\mu(P) \cap \mu(Q)$, $(I,J)$ is no witness for the satisfaction of any
of the conditions (a), (b), (c), and the algorithm rejects.

At this point, both $I,J$ are known to be in $\mu(P) \cap
\mu(Q)$. Next, compute $p=\dff^P(I,J)$ and $q=\dff^Q(I,J)$. By
Lemma~\ref{lem:ptime-diff}, this can be done in polynomial time.  If
$p\neq q$, condition (b) is satisfied and the algorithm accepts. Next,
check for 
$I>^P J$ and $I>^Q J$, 
which by
Lemma~\ref{lem:ptime-pref-check} can be done in polynomial time.
Since $I,J\in\mu(P)\cap\mu(Q)$, we have in fact checked for 
$I>^P_{\mu(P)} J$ and $I>^Q_{\mu(Q)} J$.
In case exactly one of the two holds, we accept, otherwise we
reject.  This yields a nondeterministic algorithm to decide $P\not
\equiv^{s} Q$ which runs in polynomial time (in case of ASO-problem
with access to an $\NP$-oracle), and yields the desired upper
bounds.

For hardness, we observe that in case 
of problems with empty selectors, 
$\equiv^{s}$  coincides with 
equivalence of propositional theories in case of CO-problems, 
and with 
equivalence of equilibrium theories in case of ASO-problems.
The former is well known to be $\CONP$-hard, the latter
is $\PiP{2}$-hard (Theorem~11 of \cite{PearceTW09}).
\end{proof}

\thrm{thm:co-sel:ranked}{}{
Given optimization problems $P$ and $Q$ and rank interval $[i,j]$,
deciding $P\equiv^{s,[i,j]} Q$ is $\PiP{2}$-complete in case of CO-problems 
and $\PiP{3}$-complete in case of ASO-problems.
}

\smallskip
\noindent
\begin{proof}
By Theorem~\ref{thm:sequivrankedprefeq}, $P\equiv^{s,[i,j]} Q$
iff 
jointly
(1) $\pref(P_{<i})=\pref(Q_{<i})$,
(2) ${>^P_{\pref(P_{<i})}}={>^Q_{\pref(Q_{<i})}}$, and
(3) for every $I,J\in \pref(P_{<i})$ such that $i<\dff^P(I,J)$ or
$i<\dff^Q(I,J)$, $\dff^P(I,J)=\dff^Q(I,J)$ or both $\dff^P(I,J)>j$ and
$\dff^Q(I,J)>j$.

For membership, we consider the complementary problem, that is, deciding whether
(a) $\pref(P_{<i}) \neq \pref(Q_{<i})$,
(b) ${>^P_{\pref(P_{<i})}} \neq {>^Q_{\pref(Q_{<i})}}$, or
(c) for some $I,J\in \pref(P_{<i})$ such that $i<\dff^P(I,J)$ or
$i<\dff^Q(I,J)$, $\dff^P(I,J)\neq\dff^Q(I,J)$ and it holds that $\dff^P(I,J)\leq j$ or
$\dff^Q(I,J)\leq j$.
We show that it is in $\SigmaP{2}$ for CO-problems and in $\SigmaP{3}$
for ASO-problems by showing that a witness consisting of two
interpretations can be verified in polynomial time, using an $\NP$
oracle when dealing with a CO-problem and a $\SigmaP{2}$ oracle when
dealing with an ASO-problem.

Consider a witness $(I,J)$ of two interpretations. We first test for
$I\in\pref(P_{<i})$, $I\in\pref(Q_{<i})$, $J\in\pref(P_{<i})$, and
$J\in\pref(Q_{<i})$. By Lemma~\ref{lem:conp-prefcheck-member}, this can be
done by a call to an $\NP$ oracle for CO-problems, and by
Lemma~\ref{lem:pip2-prefcheck-member} it can be done by a call to a
$\SigmaP{2}$ oracle for ASO-problems. If $I\in\pref(P_{<i}) \Delta
\pref(Q_{<i})$ or $J\in\pref(P_{<i}) \Delta \pref(Q_{<i})$ ($\Delta$
again denotes the symmetric difference), condition (a) is satisfied
and the algorithm accepts. If $I \notin \pref(P_{<i}) \cap
\pref(Q_{<i})$ or $I \notin \pref(P_{<i}) \cap \pref(Q_{<i})$, $(I,J)$
is no witness for the satisfaction of any of the conditions (a), (b),
(c), and the algorithm rejects.

Now both $I,J$ are known to be in $\pref(P_{<i}) \cap \pref(Q_{<i})$.
Next, check for $I>^P J$ and $I>^Q J$ in polynomial time (by virtue of
Lemma~\ref{lem:ptime-pref-check}), which is equivalent to checks for
$I>^P_{\pref(P_{<i})} J$ and $I>^Q_{\pref(Q_{<i})} J$.
In
case exactly one of the two holds, we accept.  Next, compute
$p=\dff^P(I,J)$ and $q=\dff^Q(I,J)$. By Lemma~\ref{lem:ptime-diff},
this can be done in polynomial time.  Then verify that $i<p$ or $i<q$,
furthermore $p\neq q$, and $p\leq j$ or $q\leq j$. If the verification
succeeds, accept, otherwise reject.
This yields a nondeterministic algorithm to decide $P\not
\equiv^{s} Q$ which runs in polynomial time (in case of ASO-problem
with access to an $\NP$-oracle), and yields the desired upper
bounds.

For the hardness part, we start with the case of CO problems. 
Therefore, we reduce the following problem to sel-equivalence:
given two propositional theories $S$ and $T$, do they possess
the same minimal models. This problem is known to 
be $\PiP{2}$-complete 
(for instance, equivalence for 
positive disjunctive programs is known to $\PiP{2}$-complete, see 
e.g.\ 
Theorem~6.15 in
\cite{efw04},
which means testing whether two propositional formulas of a particular class have the same minimal models).
The problem remains hard if $S$ and $T$ are in negation normal form 
over the same alphabet.
Given a theory $T$ we construct
a CO problem 
$P_T$ such that the elements in $\pi(P_T)$
are in a one-to-one correspondence to the minimal models of $T$.
We adapt a construction used in \cite{bnt-unpublished}.
Given $T$ (over atoms $U$),
we construct 
$$
P^g_T = T[\neg u/u'] \cup \{  u \leftrightarrow \neg u' \mid u \in U\},
$$
where $T[\neg u/u']$ stands for replacing all $\neg u$ by $u'$ in $T$,
and 
$$
P^s_T = \{ u' > u \leftarrow \mid u \in U\}.
$$
Our first observation is that each outcome of $P_T$ must be of the
form $X \cup \{y' \mid y\in U \setminus X\}$ where $X\subseteq U$. So
for each interpretation $I\subseteq U$ we can associate $I^+ = I\cup
\{y' \mid y\in U \setminus I\}$. It is clear that $I \models \neg u$ if and only if
$I^+ \models u'$, and hence also $I \models T$ if and only if $I^+
\models T[\neg u/u']$. Hence there is a one-to-one mapping between
models $M$ of $T$ and outcomes $M^+$ of $P_T$.

Now assume that $M^+ \in \pref(P_T)$. Then $M \models T$, and for any
$N\subsetneq M$ we can show that $N \not \models T$: Indeed if $N
\models T$ would hold, then $N^+ \models P^g_T$ and for rules $r\in
P^s_T$ of the form $u' > u$ where $u\in M\setminus N$ we obtain
$v_{N^+}(r) = 1 < 2 = v_{M^+}(r)$ and for all other rules $r'$ in
$P^s_T$ we have $v_{N^+}(r) = 2 = v_{M^+}(r)$, hence $N^+ >^{P_T}
M^+$, contradicting $M^+ \in \pref(P_T)$.

Assume that $M$ is a minimal model of $T$. Then $M^+ \in \mu(P_T)$ and
for all $N^+ \in \mu(P_T)$ we can show that $N^+ >^{P_T} M^+$ does not
hold (implying $M^+ \in \pref(P_T)$): If $N^+ >^{P_T} M^+$ would hold,
then $v_{N^+}(r) < v_{M^+}(r)$ for at least one $r\in P^s_T$ and
$v_{N^+}(r) \leq v_{M^+}(r)$ for all $r'\in P^s_T$. This of course
implies $N \subsetneq M$ and since $N\models T$ it contradicts the
assumption that $M$ is a minimal model of $T$.

We thus have that $M^+ \in \pref(P_T)$ if and only if $M$ is a minimal
model of $T$. Moreover, for $S$ and $T$ over $U$ it follows that $S$
and $T$ have the same minimal models if and only if $P_S
\equiv^{s,\geq 2} P_T$. Indeed, for any $R \in \L^{s,\geq 2}$ it is
easy to verify that $\pref(P_s) = \pref(P_S \cup R)$ and $\pref(P_T) =
\pref(P_T \cup R)$, since for no $I,J \in \mu(P_S)$ it holds that $I
\approx^{P_S} J$ and neither for any $I,J \in \mu(P_T)$ it holds that
$I \approx^{P_T} J$.

Concerning hardness for ASO problems, we can 
use a similar idea. However, we shall use the following 
problem: given two open QBFs
$\forall Y \phi(X,Y)$,
$\forall Y \psi(X,Y)$, do these two QBFs possess 
the same minimal models.
By Lemma~\ref{lem:compl-qbf-min-mod}, this problem is
$\PiP{3}$-hard.
We can assume that $\phi$ and $\psi$ are in negation normal form. 
The reduction then combines the idea 
from above with the 
reduction for general ASP consistency \cite{EiterG95}.
More precisely, 
we construct $P_\phi$  for a given $\phi(X,Y)$ as follows:
\begin{eqnarray*}
P^g_\phi 
&=& \{ z \vee z' \mid z\in X\cup Y\}\cup\\
&& \{ (y\wedge y')\rightarrow w, w\rightarrow y, w\rightarrow y'  \mid y\in Y\}\cup\\
&& \{ \phi[\neg z/z'] \rightarrow w,
\neg w \rightarrow w \},
\end{eqnarray*}
where $\phi[\neg z/z']$ stands for replacing all $\neg z$ by $z'$ in $\phi(X,Y)$.
For the selector we set
$$
P^s_\phi = \{ x' > x \leftarrow \mid x \in X\}.
$$ 
Any equilibrium model $M$ of $P^g_\phi$ must contain $w$ (otherwise
$\neg w \rightarrow w$ would be unsatisfied), and it must also contain
all of $\{y,y' \mid y \in Y\}$ (otherwise $w\rightarrow y$,
$w\rightarrow y'$ would be unsatisfied); let $W$ denote $\{y,y' \mid y
\in Y\} \cup \{w\}$, the set contained in each equilibrium
model. Moreover, each equilibrium model $M$ must be of the form $V
\cup \{z' \mid z \in X \setminus V\} \cup W$: clearly one of $x$ and
$x'$ must hold for each $x \in X$ to satisfy $v \vee v'$, but not
both, as otherwise $\langle M\setminus \{v\},M\rangle \models_{HT}
P^g_\phi$ as well. We can therefore associate each interpretation $I
\subseteq X$ for $\forall Y\phi(X,Y)$ to exactly one potential
equilibrium model $I^+ = I \cup \{x' \mid x \in X \setminus I\} \cup
W$. Moreover, we can show that $I$ satisfies $\forall Y\phi(X,Y)$ if
and only if $I^+$ is an equilibrium model of $P^g_\phi$: If $I$
satisfies $\forall Y\phi(X,Y)$, then clearly $\langle I^+,I^+ \rangle
\models_{HT} P^g_\phi$ and for all $J \subsetneq I^+$ it holds that
$\langle J,I^+ \rangle \not \models_{HT} P^g_\phi$; if $I$ does not
satisfy $\forall Y\phi(X,Y)$ then there exists some $J \subseteq Y$
such that $I\cup J \not \models \phi(X,Y)$, therefore even if $\langle
I^+,I^+ \rangle \models_{HT} P^g_\phi$, also $\langle I^+ \setminus
(\{y \mid y \notin J\} \cup \{y' \mid y \in J\} \cup \{w\}) ,I^+
\rangle \models_{HT} P^g_\phi$ (the first component of the HT
interpretation is of the form $I \cup \{x' \mid x \in X \setminus I\}
\cup J \cup \{y' \mid y \in Y \setminus J\}$ and hence does not
satisfy $\phi[\neg z/z']$) and hence in this case $I^+$ is not an
equilibrium model. So there is a one-to-one mapping between models $M$ of $\forall Y\phi(X,Y)$ and outcomes $M^+$ of $P_\phi$.

Once this is established, we can reason as for CO-problems concerning
minimality. Assume first that $M^+ \in \pref(P_\phi)$. Then $M$
satisfies $\forall Y\phi(X,Y)$, and for any $N\subsetneq M$ we can
show that $N$ does not satisfy $\forall Y\phi(X,Y)$: Indeed if $N$
would satisfy $\forall Y\phi(X,Y)$, then $N^+ \in \mu(P_\phi)$ and for
rules $r\in P^s_\phi$ of the form $u' > u$ where $u\in M\setminus N$
we obtain $v_{N^+}(r) = 1 < 2 = v_{M^+}(r)$ and for all other rules
$r'$ in $P^s_\phi$ we have $v_{N^+}(r) = 2 = v_{M^+}(r)$, hence $N^+
>^{P_\phi} M^+$, contradicting $M^+ \in \pref(P_\phi)$.  Assume next
that $M$ is a minimal model of $\forall Y\phi(X,Y)$. Then $M^+ \in
\mu(P_\phi)$ and for all $N^+ \in \mu(P_\phi)$ we can show that $N^+
>^{P_\phi} M^+$ does not hold (implying $M^+ \in \pref(P_\phi)$): If
$N^+ >^{P_\phi} M^+$ would hold, then $v_{N^+}(r) < v_{M^+}(r)$ for at
least one $r\in P^s_\phi$ and $v_{N^+}(r) \leq v_{M^+}(r)$ for all
$r'\in P^s_\phi$. This of course implies $N \subsetneq M$ and since
$N$ satisfies $\forall Y\phi(X,Y)$ it contradicts the assumption that
$M$ is a minimal model of $\forall Y\phi(X,Y)$. We thus have that $M^+
\in \pref(P_\phi)$ if and only if $M$ is a minimal model of $\forall
Y\phi(X,Y)$.  Moreover, for $\phi$ and $\psi$ over $X\cup Y$ it
follows that $\forall Y\phi(X,Y)$ and $\forall Y\psi(X,Y)$ have the
same minimal models if and only if $P_\phi \equiv^{s,\geq 2}
P_\psi$. Indeed, for any $R \in \L^{s,\geq 2}$ it is easy to verify
that $\pref(P_\phi) = \pref(P_\phi \cup R)$ and $\pref(P_\psi) =
\pref(P_\psi \cup R)$, since for no $I,J \in \mu(P_\phi)$ it holds
that $I \approx^{P_\phi} J$ and neither for any $I,J \in \mu(P_\psi)$
it holds that $I \approx^{P_\psi} J$.
\end{proof}

\thrm{thm:co-gen}{}{
Given two CO (ASO, respectively) problems $P$ and $Q$, deciding $P \equiv_g Q$
is $\CONP$-complete.
}

\smallskip
\noindent
\begin{proof}
By Theorem~\ref{thm:simpleg}, $P \equiv_g Q$ if and only if
${>^P_{{\Mod{P^g}}}}= {>^Q_{{\Mod{Q^g}}}}$ and for CO problems
$\Mod{P^g}=\Mod{Q^g}$ and for ASO problems $\HT{P^g} =
\HT{Q^g}$. Hardness immediately follows from $\CONP$-completeness of
testing $\Mod{P^g}=\Mod{Q^g}$ and $\HT{P^g} = \HT{Q^g}$
\cite{lin-2002-kr}.  Membership can be shown by verifying a witness
$(I,J)$ for the complementary problem (composed of two
interpretations) in polynomial time: Check in polynomial time whether
$I \in \Mod{P^g}$, $I \in \Mod{Q^g}$, $J \in \Mod{P^g}$, $J \in
\Mod{Q^g}$. If all checks succeed, check for $I >^P_{\Mod{P^g}} J$ and
$I >^Q_{\Mod{Q^g}} J$ (by Lemma~\ref{lem:ptime-pref-check} 
$I >^P J$ and $I >^Q J$ can be checked in
polynomial time
and we have $I,J\in\Mod{P^g}$ and $I,J\in\Mod{Q^g}$); 
if one of them succeeds and the other does not,
accept; otherwise, for CO problems test for $I \in \Mod{P^g} \Delta
\Mod{Q^g}$ or $J \in \Mod{P^g} \Delta \Mod{Q^g}$ ($\Delta$ again
denotes the symmetric difference) and accept if so, reject otherwise;
for ASO problems, check for $\langle I,J\rangle \models_{HT} P^g$ and
$\langle I,J\rangle \models_{HT} Q^g$ (in polynomial time); if one of
them succeeds and the other does not, accept, reject otherwise.
\end{proof}

\thrm{thm:co-combined}{}{
Given ranked CO (ASO, respectively) problems $P$ and $Q$,
and rank interval $[i,j]$,
deciding $P \equiv^{s,[i,j]}_g Q$ is $\CONP$-complete.
}

\smallskip
\noindent
\begin{proof}
By Theorem~\ref{thm:cequivrankedprefeqcomb},
$P \equiv^{s,[i,j]}_g Q$ if and only if jointly
(1) $P^g$ and $Q^g$ are strongly equivalent (that is, $\Mod{P^g}=
\Mod{Q^g}$ for CO problems, and $\HT{P^g}=\HT{Q^g}$ for ASO problems),
(2) ${>^P_{\Mod{P^g}}}={>^Q_{\Mod{Q^g}}}$,
(3) for every $I,J\in \Md(P^g)$ such that $i<\dff^P(I,J)$ or
$i<\dff^Q(I,J)$, $\dff^P(I,J)=\dff^Q(I,J)$ or both $\dff^P(I,J)>j$ and
$\dff^Q(I,J)>j$, and
(4) ${>^{P_{<i}}_{\Mod{P^g}}}= {>^{Q_{<i}}_{\Mod{Q^g}}}$.

For membership, we consider the complementary problem, that is, deciding whether
(a) $P^g$ and $Q^g$ are not strongly equivalent (that is, $\Mod{P^g} \neq
\Mod{Q^g}$ for CO problems, and $\HT{P^g} \neq \HT{Q^g}$ for ASO problems), or
(b) ${>^P_{\Mod{P^g}}} \neq {>^Q_{\Mod{Q^g}}}$,
(c) for some $I,J\in \Md(P^g)$ such that $i<\dff^P(I,J)$ or
$i<\dff^Q(I,J)$, $\dff^P(I,J) \neq \dff^Q(I,J)$ and it holds that $\dff^P(I,J)\leq j$ or
$\dff^Q(I,J)\leq j$, and
(d) ${>^{P_{<i}}_{\Mod{P^g}}}\neq {>^{Q_{<i}}_{\Mod{Q^g}}}$.
As in previous proofs, we will give a polynomial time algorithm for
verifying a witness consisting of two interpretations.

Consider therefore such a witness $(I,J)$ of two interpretations. We
first test for $I\in \Mod{P^g}$, $I\in \Mod{Q^g}$, $J\in \Mod{P^g}$,
and $J\in \Mod{Q^g}$. For CO problems, if $I\in \Mod{P^g} \Delta
\Mod{Q^g}$ or $J\in \Mod{P^g} \Delta \Mod{Q^g}$ (where $\Delta$ again
denotes the symmetric difference), then accept immediately. For ASO
problems, test for $\langle I,J \rangle \models_{HT} P^g$ and $\langle
I,J \rangle \models_{HT} Q^g$ (in polynomial time), and if only one of
them holds, then accept. If $I \notin \Mod{P^g} \cap \Mod{Q^g}$ or
$J\in \Mod{P^g} \cap \Mod{Q^g}$, reject.

At this point, both $I,J$ are known to be in $\Mod{P^g} \cap
\Mod{Q^g}$. Thus, checking for $I>^P_{\Mod{P^g}} J$ and $I>^Q_{\Mod{Q^g}}
J$ 
can be done in polynomial time, c.f.\ 
Lemma~\ref{lem:ptime-pref-check}. 
In case exactly one of the two holds, we
accept. Next, compute $p=\dff^P(I,J)$ and $q=\dff^Q(I,J)$. By
Lemma~\ref{lem:ptime-diff}, this can be done in polynomial time. Then
verify that $i<p$ or $i<q$, furthermore $p\neq q$, and $p\leq j$ or
$q\leq j$. If the verification succeeds, accept. Finally, check for
$I>^{P_{<i}}_{\Mod{P^g}} J$ and $I>^{Q_{<i}}_{\Mod{Q^g}} J$, which by
Lemma~\ref{lem:ptime-pref-check} can be done in polynomial time. In
case exactly one of the two holds, we accept, otherwise reject.

Hardness follows directly from $\CONP$-completeness of deciding
equivalence between two propositional theories and of deciding strong
equivalence between two equilibrium theories \cite{lin-2002-kr}.
\end{proof}

}

\end{document}